\documentclass[ a4paper
              , 11pt 
              ]{amsart}

\DeclareMathAlphabet{\mybb}{U}{BOONDOX-ds}{m}{n}
\usepackage[all,2cell]{xy}\UseAllTwocells

\usepackage{hyperref}
\usepackage{booktabs
           , subcaption
           , adjustbox
           , graphicx
           , amsmath
           , amssymb
           , tikz-cd
           , todonotes
           , stmaryrd
           , enumerate
           , xparse  % for latex3 fast macro definition
           , xspace
           , mathpartir
           , microtype
           , newtxmath
           , newtxtext
           }
\usepackage[shortlabels]{enumitem}
\setlist[enumerate]{label=\oldstylenums{\arabic*})}

\hypersetup{ colorlinks
           , linkcolor=blue!70!black
           , urlcolor=black
           , citecolor=red!70!black
           , bookmarks=true
           , pdftitle={Functorial Semantics for Partial Theories}}

\usepackage[left=3cm, right=3cm,top=3cm,bottom=4.5cm]{geometry}

\usepackage{stmaryrd}

\newcommand{\semfun}[1]{\llbracket#1\rrbracket}

\newcommand{\parcat}[1]{\mathsf{Par}(#1)}
\def\unp{Universal property}
\def\als{Algebraic structure}
\def\fre{Free c. on 1}
\def\syn{Syntax}
\def\sem{Semantics}
\def\spe{Specification}

\newcommand{\Par}{\mathsf{Par}}

\newcommand{\partialto}{\rightharpoonup}
\newcommand{\dom}[1]{\mathsf{dom}#1}
\newcommand{\unf}[1]{\mathsf{def}#1}
\newcommand{\rest}[1]{\overline{#1}}

\newcommand{\ra}{\to}
\def\To{\Rightarrow}
\newcommand{\Defeq}{\stackrel{\text{def}}{=}}
\newcommand{\eed}{\text{id}}
\newcommand{\from}{\colon}
\newcommand{\substitute}[2]{\raise2pt\hbox{$\scriptstyle{#1}$}/\lower2pt\hbox{$\scriptstyle{#2}$}}

\def\ev{\text{ev}}

\ExplSyntaxOn
\NewDocumentCommand{\makeabbrev}{mmm}
 {
  \yoruk_makeabbrev:nnn { #1 } { #2 } { #3 }
 }

\cs_new_protected:Npn \yoruk_makeabbrev:nnn #1 #2 #3
 {
  \clist_map_inline:nn { #3 }
   {
    \cs_new_protected:cpn { #2 } { #1 { ##1 } }
   }
 }
\ExplSyntaxOff

\makeabbrev{\textbf}{bf#1}{
  a,b,c,d,e,f,g,h,i,j,k,l,m,n,o,p,q,r,t,u,v,w,x,y,z,%
  A,B,C,D,E,F,G,H,I,J,K,L,M,N,O,P,Q,R,T,U,V,W,X,Y,Z }
\makeabbrev{\boldsymbol}{bs#1}{%
    a,b,c,d,e,f,g,h,i,j,k,l,m,n,o,p,q,r,s,t,u,v,w,x,y,z,%
    A,B,C,D,E,F,G,H,I,J,K,L,M,N,O,P,Q,R,S,T,U,V,W,X,Y,Z }
\makeabbrev{\mathsf}{sf#1}{
  a,b,c,d,e,f,g,h,i,j,k,l,m,n,o,p,q,r,s,t,u,v,w,x,y,z,%
  A,B,C,D,E,F,G,H,I,J,K,L,M,N,O,P,Q,R,S,T,U,V,W,X,Y,Z }
\makeabbrev{\mathfrak}{fk#1}{
  a,b,c,d,e,f,g,h,j,k,i,l,m,n,o,p,q,r,s,t,u,v,w,x,y,z,%
  A,B,C,D,E,F,G,H,I,J,K,L,M,N,O,P,Q,R,S,T,U,V,W,X,Y,Z }
\makeabbrev{\mathcal}{cl#1}{
  A,B,C,D,E,F,G,H,I,J,K,L,M,N,O,P,Q,R,S,T,U,V,W,X,Y,Z }
\makeabbrev{\mybb}{bb#1}{
  A,B,C,D,E,F,G,H,I,J,K,L,M,N,O,P,Q,R,S,T,U,V,W,X,Y,Z }

\newcommand{\Mod}{\ct{Mod}}

\def\ct#1{\mathsf{#1}}
\def\Set{\ct{Set}}

\def\pLaw{\ct{pLaw}}
\def\Cat{\ct{Cat}}

\def\Lex{\ct{Lex}}
\def\ssLFP{\ct{LFP}}

\def\Tho{\ct{Th}}
\def\Th{\Tho}

\def\op{\text{op}}

\def\F{\bbF}
\def\N{\mathbb{N}}

\newcommand{\bnfEq}{\; ::= \;}
\newcommand{\bnfSep}{\;\; | \;\;}

\newlength{\seplen}
\newlength{\sepwid}
\def\firstblank{\,\rule{\seplen}{\sepwid}\,}

\newcommand*{\Scale}[2][4]{\scalebox{#1}{\ensuremath{#2}}}%

\title
  {Functorial Semantics for Partial Theories}
\usetikzlibrary{matrix,arrows}
\tikzset{commutative diagrams/.cd,arrow style=tikz,diagrams={>=stealth'}}

\let\phi\varphi

\usepackage{prooftree}
\newcommand{\reductionRule}[2]{{\prooftree{\scriptstyle #1}\justifies{\scriptstyle #2}\endprooftree}}
\newcommand{\typ}{\mathrel{:}}
\newcommand{\typeJudgment}[3]{{ {#2} \,\typ\, {#3}}}
\newcommand{\sort}[2]{\ensuremath{(#1,\,#2)}}

\usepackage{tikz}
\usetikzlibrary{arrows,backgrounds,shapes.geometric,shapes.misc,matrix,arrows.meta,decorations.markings}

\pgfdeclarelayer{background}
\pgfdeclarelayer{edgelayer}
\pgfdeclarelayer{nodelayer}
\pgfsetlayers{background,edgelayer,nodelayer,main}
\tikzset{baseline=-0.5ex}
\definecolor{light-gray}{gray}{.7}
\tikzstyle{none}=[inner sep=0pt]
\tikzstyle{plain}=[inner sep=0pt]
\tikzstyle{black}=[circle, draw=black, fill=black, inner sep=0pt, minimum size=3.5pt]
\tikzstyle{black-faded}=[circle, draw=light-gray, fill=light-gray, inner sep=0pt, minimum size=4pt]
\tikzstyle{white}=[circle, draw=black, fill=white, inner sep=0pt, minimum size=3.5pt]
\tikzstyle{white-faded}=[circle, draw=light-gray, fill=white, inner sep=0pt, minimum size=4.5pt]
\tikzstyle{delay}=[fill=black, regular polygon, regular polygon sides=3,rotate=-90, scale=.55]
\tikzstyle{delay-op}=[fill=black, regular polygon, regular polygon sides=3,rotate=90, scale=.55]
\tikzstyle{reg}=[draw, fill=white, rounded rectangle, rounded rectangle left arc=none, minimum height=1em, minimum width=1em, node font={\scriptsize}]
\tikzstyle{coreg}=[draw, fill=white, rounded rectangle, rounded rectangle right arc=none, minimum height=1em, minimum width=1em, node font={\scriptsize}]
\tikzstyle{rn}=[circle, draw=red, fill=red, inner sep=0pt, minimum size=4pt]
\tikzstyle{place}=[circle, draw=black, fill=white, inner sep=0pt, minimum size=8pt]

\newcommand{\symDiag}{
\tikzset{x=1em, y=2.1ex}
\begin{tikzpicture}
	\begin{pgfonlayer}{nodelayer}
		\node [style=none] (0) at (-0.75, -0.5) {};
		\node [style=none] (1) at (-0.75, 0.5) {};
		\node [style=none] (2) at (0.75, -0.5) {};
		\node [style=none] (3) at (0.75, 0.5) {};
	\end{pgfonlayer}
	\begin{pgfonlayer}{edgelayer}
		\draw [in=180, out=0] (1.center) to (2.center);
		\draw [in=180, out=0] (0.center) to (3.center);
	\end{pgfonlayer}
\end{tikzpicture}
}
\tikzset{x=1em, y=1.5ex}
}

\newcommand{\comp}{\mathop{\fatsemi}}%{\mathrel{;}}

\def\sSigma{\mathfrak{S}}

\DeclareMathOperator{\colim}{colim}
\def\L{L}

\DeclareDocumentEnvironment{enumtag}{m}{%
	\begin{enumerate}%
		[ label = \textsc{#1}\oldstylenums{\arabic*}),%
		  ref   = \textsc{#1}\oldstylenums{\arabic*}%
	  ] }{ \end{enumerate} }

  %% \keywords are mandatory in final camera-ready submission

\newcommand{\graffle}[3][0pt]{\lower#1\hbox{$\includegraphics[height=#2]{graffles/#3.pdf}$}}
\def\circ{\cdot}

\usetikzlibrary{calc}
\def\len{0.25}

\tikzset{
  wire/.style={rounded corners=3, line width=.5pt}
, iden/.style={line width=.5pt}
, dot/.style={draw,fill=black, circle, inner sep=1pt}
, wdot/.style={draw,fill=white, circle, inner sep=1pt}
, edot/.style={draw,fill=red!40, circle, inner sep=1pt}
, cdot/.style={draw,fill=#1, circle, inner sep=1pt}
, labeled/.style={draw,fill=white, inner sep=1pt}
}

\def\akasa{
\draw[densely dotted] (0,0) rectangle (1,1);
}

\def\braid{
  \draw[draw=none] (0,0) rectangle (1,1);
\draw[wire] (0,\len) -- (\len,\len) -- (3*\len, 3*\len) -- (1,3*\len);
\draw[white,wire] (0,3*\len) -- (\len,3*\len) -- (3*\len, \len) -- (1,\len);
}
\def\id{
  \draw[draw=none] (0,0) rectangle (1,1);
\draw[iden] (0,2*\len) -- (1,2*\len);
}
\def\twoid{
\draw[iden] (0,\len) -- (1,\len);
\draw[iden] (0,3*\len) -- (1,3*\len);
}
\def\mult{
  \draw[draw=none] (0,0) rectangle (1,1);
\draw[wire] (0,\len) -| (2*\len,3*\len) -- (0,3*\len);
\draw[wire] (2*\len,2*\len) -- (1,2*\len);
\node[dot] at (2*\len,2*\len) {};
}

\def\comult{
  \draw[draw=none] (0,0) rectangle (1,1);
\draw[wire] (4*\len,\len) -| (2*\len,3*\len) -- (1,3*\len);
\draw[wire] (2*\len,2*\len) -- (0,2*\len);
\node[dot] at (2*\len,2*\len) {};
}

\def\emult{
\draw[wire] (0,\len) -| (2*\len,3*\len) -- (0,3*\len);
\draw[wire] (2*\len,2*\len) -- (1,2*\len);
\node[edot] at (2*\len,2*\len) {};
}
\def\ecomult{
\draw[wire] (4*\len,\len) -| (2*\len,3*\len) -- (1,3*\len);
\draw[wire] (2*\len,2*\len) -- (0,2*\len);
\node[edot] at (2*\len,2*\len) {};
}
\def\unit{
  \draw[draw=none] (0,0) rectangle (1,1);
\draw[wire] (2*\len,2*\len) -- (1,2*\len);
\node[dot] at (2*\len,2*\len) {};
}
\def\eunit{
\draw[wire] (2*\len,2*\len) -- (1,2*\len);
\node[edot] at (2*\len,2*\len) {};
}
\def\lid{
  \draw[draw=none] (0,0) rectangle (1,1);
\draw[iden] (0,\len) -- (1,\len);
}
\def\uid{
  \draw[draw=none] (0,0) rectangle (1,1);
\draw[iden] (0,3*\len) -- (1,3*\len);
}
\def\mor#1{
  \draw[draw=none] (0,0) rectangle (1,1);
\id
\node[draw, fill=white, inner sep=1.5pt] at (2*\len,2*\len) {\tiny $#1$};
}
\def\lmor#1{
\lid
\node[draw, fill=white, inner sep=1.5pt] at (2*\len,\len) {\tiny $#1$};
}
\def\umor#1{
\uid
\node[draw, fill=white, inner sep=1.5pt] at (2*\len,3*\len) {\tiny $#1$};
}
\def\dcouni#1{
  \draw[draw=none] (0,0) rectangle (1,1);
\draw[wire] (2*\len,2*\len) -- ++(2*\len,0);
\node[draw, fill=white, inner sep=1.5pt] at (2*\len,2*\len) {\tiny $#1$};
}
\def\ldcouni#1{
\draw[wire] (2*\len,\len) -- ++(2*\len,0);
\node[draw, fill=white, inner sep=1.5pt] at (2*\len,\len) {\tiny $#1$};
}

\def\lcomult{
\draw[wire] (4*\len,0) -| (2*\len,2*\len) -- (1,2*\len);
\draw[wire] (2*\len,\len) -- (0,\len);
\node[dot] at (2*\len,\len) {};
}

\def\ucomult{
\draw[wire] (4*\len,2*\len) -| (2*\len,1) -- (1,1);
\draw[wire] (2*\len,3*\len) -- (0,3*\len);
\node[dot] at (2*\len,3*\len) {};
}
\def\counit{
\draw[wire] (0,2*\len) -- (2*\len,2*\len);
\node[dot] at (2*\len,2*\len) {};
}

\def\ucounit{
\draw[wire] (0,3*\len) -- (2*\len,3*\len);
\node[dot] at (2*\len,3*\len) {};
}
\def\lcounit{
\draw[wire] (0,\len) -- (2*\len,\len);
\node[dot] at (2*\len,\len) {};
}
\def\thingy#1{
  \draw[wire] (4*\len,\len) -| (2*\len,3*\len) -- (1,3*\len);
  \draw[wire] (2*\len,2*\len) -- (0,2*\len);
  \draw[fill=white] (\len,\len-.1) rectangle (3*\len,3*\len+.1) node[pos=.5] {$\scriptscriptstyle #1$};
}

\def\turn{
	\draw[rounded corners=3, line width=.5pt] (0,\len) -- (\len,\len) -- (3*\len, 3*\len) -- (4*\len,3*\len);
}
\def\coturn{
	\draw[rounded corners=3, line width=.5pt] (0,3*\len) -- (\len,3*\len) -- (3*\len, \len) -- (4*\len,\len);
}

\def\tenmult{
\draw[wire] (0,\len) -| (2*\len,3*\len) -- (0,3*\len);
\draw[wire] (2*\len,2*\len) -- (1,2*\len);
\node[fill=white,inner sep=-1.25pt,circle] at (2*\len,2*\len) { $\otimes$};
}

\newcommand{\ezs}[2]{
\begin{scope}[#2]
#1
\end{scope}
}

\renewcommand\fbox{\fcolorbox{gray!30}{white}}

\newcommand{\step}[2][1]{\ezs{#2}{xshift=#1 cm}}

\newcommand{\leftLabel}[2][0]{
\node[left] at (#1,2*\len) {$#2 =\quad$};
}
\newcommand{\rightLabel}[2]{
\node[right] at (#2,2*\len) {$\quad = #1$};
}

\newcommand{\umult}[1][black]{
\ezs{
\draw[wire] (0,\len) -| (2*\len,3*\len) -- (0,3*\len);
\draw[wire] (2*\len,2*\len) -- (1,2*\len);
\node[cdot=#1] at (2*\len,2*\len) {};
}{yshift=\len cm}
}

\newcommand{\lmult}[1][black]{
\ezs{
\draw[wire] (0,\len) -| (2*\len,3*\len) -- (0,3*\len);
\draw[wire] (2*\len,2*\len) -- (1,2*\len);
\node[cdot=#1] at (2*\len,2*\len) {};
}{yshift=-\len cm}
}

\def\cothingy#1{
  \ezs{\thingy{#1}}{xscale=-1}
}

\def\multi#1{
\draw[wire] (0,5*\len) -| (2*\len,7*\len) -- (0,7*\len);
\draw[wire] (0,\len) -| (2*\len,3*\len) -- (0,3*\len);
\ezs{\draw[wire] (3*\len,2*\len) -- ++(\len,0);
\draw[fill=white, thin] (\len,\len-.15) rectangle (3*\len,3*\len+.15) node[pos=.5,] {$\scriptscriptstyle #1$};}{yscale=2}
}

\newcommand{\up}[2][1]{
    {\ezs{#2}{yshift=#1*\len cm}}
}
\newcommand{\down}[2][1]{
    {\ezs{#2}{yshift=-#1*\len cm}}
}

\newcommand{\premor}[3]{
  \draw[iden] (0,2*\len) -- (1,2*\len);
  \node[draw, fill=white, inner sep=1.5pt] at (2*\len,2*\len) {\tiny $#1$};
\node[left] at (0,2*\len) {\tiny $#2$};
\node[right] at (1,2*\len) {\tiny $#3$};
}

\newcommand{\tmor}[3]{
\mor{#1}
\node[left] at (0,2*\len) {\tiny $#2$};
\node[right] at (1,2*\len) {\tiny $#3$};
}

\newcommand{\gau}[1]{
\node[left] at (0,2*\len) {\tiny $#1$};
}

\newcommand{\dro}[1]{
\node[right] at (1,2*\len) {\tiny $#1$};
}

\newcommand{\eql}[1]{
\node at (#1,2*\len) {$=$};
}

\def\theR{\ucounit
\lcounit
\draw[fill=white] (\len,\len-.1) rectangle (3*\len,3*\len+.1) node[pos=.5] {\tiny $R$};}

\theoremstyle{definition}
 \newtheorem{theorem}{Theorem}[section]
 \newtheorem*{theorem*}{Theorem}
 \newtheorem*{definition*}{Definition}
  \newtheorem{remark}[theorem]{Remark}
  \newtheorem{notation}[theorem]{Notation}
  \newtheorem{observation}[theorem]{Observation}
  \newtheorem{recipe}[theorem]{Recipe}
  
  \newtheorem{proposition}[theorem]{Proposition}
  \newtheorem{lemma}[theorem]{Lemma}
  \newtheorem{corollary}[theorem]{Corollary}
  \newtheorem{example}[theorem]{Example}
  \newtheorem{definition}[theorem]{Definition}

\begin{document}
\allowdisplaybreaks
\author{Ivan Di Liberti}
\email{diliberti.math@gmail.com}
\address{
  % \textbf{Ivan Di Liberti}\newline
  Institute of Mathematics, \newline
  Czech Academy of Sciences,\newline
  Prague, Czech Republic.
}

\author{Fosco Loregian}
\author{Chad Nester}
\author{Pawe{\l} Soboci\'nski}
\email{fosco.loregian@taltech.ee}
\email{chad.nester@taltech.ee}
\email{pawel.sobocinski@taltech.ee}
\address{
  Department of Software Science,\newline
  Tallinn University of Technology,\newline
  Tallinn, Estonia.
}

\begin{abstract}
We provide a Lawvere-style definition for partial theories, extending the classical notion of equational theory by allowing partially defined operations. As in the classical case, our definition is syntactic: we use an appropriate class of string diagrams as terms. This allows for equational reasoning about the class of models defined by a partial theory. We demonstrate the expressivity of such equational theories by considering a number of examples, including partial combinatory algebras and cartesian closed categories. Moreover, despite the increase in expressivity of the syntax we retain a well-behaved notion of semantics: we show that our categories of models are precisely locally finitely presentable categories, and that free models exist.
\end{abstract}

\maketitle
\paragraph{\bf Acknowledgements}
  Di Liberti was supported by the Grant Agency of the Czech Republic project EXPRO 20-31529X and RVO: 67985840. Loregian, Nester and Soboci\'{n}ski were supported by the ESF funded Estonian IT Academy research
  measure (project 2014-2020.4.05.19-0001).

\section{Introduction}
\def\fix#1{\textcolor{black}{#1}}

% equational theories + equational reasoning
Mathematicians interested in, say, the theory of monoids or the theory of groups work in an axiomatic setting, asserting the presence of a collection of $n$-ary \emph{operations} on an ambient set $A$ --- i.e.\ (total) functions $A^n\ra A$ for some $n:\bbN$ --- that satisfy a number of axioms. This data can be packaged up into an \emph{equational theory}: a pair $(\Sigma,E)$ where $\Sigma$ is the \emph{signature}, consisting of \emph{operation symbols}, each with a specified arity, and $E$ is a collection of \emph{equations} --- i.e.\ pairs of \emph{terms} built up from the signature $\Sigma$ and auxiliary variables --- that provide the axioms.
An ambient equational theory is thus the bread and butter of an algebraist, that together with the principles of equational reasoning provides the basic calculus of mathematical investigation into the structure of interest.

% universal algebra + varieties
Birkhoff~\cite{birkhoff1935structure} discovered that a substantial amount of mathematics can be done at the level of generality of an equational theory. Given an equational theory $(\Sigma,E)$, a \emph{model} is a set together with an interpretation of the function symbols $\Sigma$ that satisfies the equations $E$. A monoid is nothing but a model of the equational theory of monoids, a group is a model of the equational theory of groups, and so on. The semantics of an equational theory, i.e.\ its class of models, is called a \emph{variety}. %Working with generic varieties,
Birkhoff showed that certain results (e.g.\ the so-called isomorphism theorems, existence of free models) can be derived uniformly for generic varieties, independent of the equational theory at hand. %Moreover, there are structural similarities between the classes of models for different.
Most spectacularly, a class of sets-with-structure can be determined to be a variety through purely structural means; this is often referred to as Birkhoff's Variety Theorem or the HSP Theorem.

The resulting field is known as \emph{universal algebra}. Its mathematical objects of study are equational theories and varieties. Given its goal of uncovering methodological and technical similarities of a large swathe of contemporary algebra, universal algebra is in the intersection of mathematics and mathematical logic. It has influenced computer science, especially programming language theory, as a formal and generic treatment of syntax, terms, equational reasoning, etc.

\smallskip

% functorial semantics
Lawvere~\cite{lawvere1963functorial}, and the subsequent development of \emph{categorical} universal algebra, addressed some of the perceived shortcomings of the classical account. It is well known that a single variety can have many different axiomatic presentations, and in this sense the choice of a particular presentation may seem ad hoc. The requirement that models be sets-with-structure is also  restrictive, since algebraic structures appear in other mathematical contexts as well.
A Lawvere theory is a category $\clL$ that serves as a \emph{presentation-independent} way of capturing the specification of a variety. A central conceptual role is played by \emph{cartesian categories}, i.e.\ categories with finite products. The free cartesian category on one object often appears in the very definition of a Lawvere theory  -- the ``one object'' here capturing single-sortedness. Finite products track arities and ensure that operations are total functions. Functorial semantics gives us the correct generalisation of varieties: a model is cartesian functor $\clL\to\Set$. This point of view is flexible (e.g.\ $\Set$ can be replaced with another cartesian category) and leads to a rich theory~\cite{hyland2007category,lawvere1963functorial,adamek2003duality}, where the study of varieties and their specifications can take place at a high level of generality.

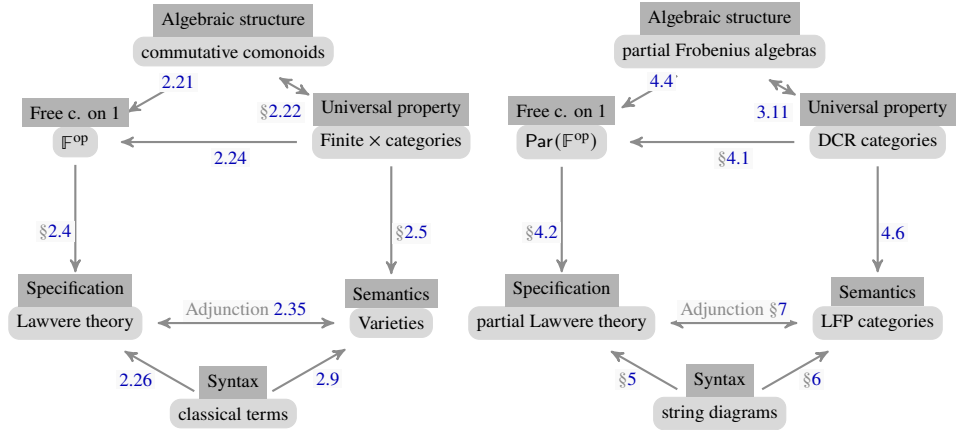
\begin{figure}
	\begin{center}
		\begin{tikzpicture}[scale=.8,
				cat/.style={
						draw=gray!30,
						fill=gray!30,
						rectangle,
						rounded corners,
						thin,
						font=\tiny,
						inner sep=3pt
					},
				chapter number/.style={
						anchor=south,
						fill=gray!60,
						font=\tiny,
						inner sep=3pt,
						outer sep=0pt,
						minimum size=.15in,
					},
				impli/.style={
						thick
						, gray!90
						, >=stealth'
						, ->
					},
				imply node/.style={
						fill=gray!5,
						font=\tiny,
						inner sep=1pt
					}
			]
			\foreach \i/\j in {
			0/{commutative comonoids}
			, 1/{$\bbF^\text{op}$}
			, 2/{Lawvere theory}
			, 3/{classical terms}
			, 4/{Varieties}
			, 5/{Finite $\times$ categories}
			}{
			\node[cat] (\i) at (90+60*\i:3) {\j};
			}
			\draw[impli]
			(1) edge[shorten >=4.5mm,shorten <=1mm] node[imply node, left] {§\ref{subsec:law}} (2)
			(3) edge[shorten >=3mm,shorten <=2mm] node[imply node, below left] {\ref{rem:terms}} (2)
			edge[shorten >=3mm,shorten <=2mm] node[imply node, below right] {\ref{birko}} (4)
			(5) edge[shorten >=4.5mm,shorten <=1mm] node[imply node, right] {§\ref{sem_for_alg}} (4)
			(4) edge[<->,shorten >=2mm,shorten <=2mm] node[imply node, above] {Adjunction \ref{thm:lawvereadjunction}} (2)
			(5) edge[shorten >=3mm,shorten <=8mm, <->] node[imply node, below left=1mm] {§\ref{thm:fox}} (0)
			edge[shorten >=3mm,shorten <=2mm] node[imply node, below right=1mm] {\ref{Fop_is_free_on_1}} (1)
			(0) edge[shorten >=4.5mm,shorten <=2mm] node[imply node, above right=1mm] {\ref{obs:Fopcomonoids}} (1);
			\foreach \u/\v in {
					0/\als
					, 1/\fre
					, 2/\spe
					, 3/\syn
					, 4/\sem
					, 5/\unp
				}{
					\node[chapter number] at (\u.90) {\v};
				}
			\begin{scope}[xshift=8cm]
				\foreach \i/\j in {
				0/{partial Frobenius algebras}
				, 1/{$\Par(\bbF^\text{op})$}
				, 2/{partial Lawvere theory}
				, 3/{string diagrams}
				, 4/{LFP categories}
				, 5/{DCR categories}
				}{
				\node[cat] (\i) at (90+60*\i:3) {\j};
				}
				\draw[impli]
				(1) edge[shorten >=4.5mm,shorten <=1mm] node[imply node, left] {§\ref{subsec:partialLawvere}} (2)
				(3) edge[shorten >=3mm,shorten <=2mm] node[imply node, below left] {§\ref{sec:partialequationaltheories}} (2)
				edge[shorten >=3mm,shorten <=2mm] node[imply node, below right] {§\ref{sec:multi}} (4)
				(5) edge[shorten >=4.5mm,shorten <=1mm] node[imply node, right] {\ref{models_of_parlaw}} (4)
				(4) edge[<->,shorten >=2mm,shorten <=2mm] node[imply node, above] {Adjunction §\ref{the_variety}} (2)
				(5) edge[shorten >=3mm,shorten <=8mm, <->] node[imply node, below left=1mm] {\ref{thm:dcrc}} (0)
				edge[shorten >=3mm,shorten <=2mm] node[imply node, below right=1mm] {§\ref{the_free_dcr_on}} (1)
				(0) edge[shorten >=4.5mm,shorten <=2mm] node[imply node, above right=1mm] {\ref{as_props_ultimo}} (1);
				\foreach \u/\v in {
						0/\als
						, 1/\fre
						, 2/\spe
						, 3/\syn
						, 4/\sem
						, 5/\unp
					}{
						\node[chapter number] at (\u.90) {\v};
					}
			\end{scope}
		\end{tikzpicture}
	\end{center}
	\caption{\label{fig:hexagons}Elements of classical functorial semantics on the left, and our contribution on the right.}
\end{figure}

% post-hoc
The beautiful abstract picture painted by Lawvere can be used to give a post-hoc explanation of the elements of \emph{classical} equational theories. Every equational theory yields a Lawvere theory. Free equational theories, i.e.\ those where $E=\varnothing$, are Lawvere theories whose arrows can be concretely described as (tuples of) terms. \fix{Indeed, it is well-known that terms are closely connected to the universal property of products}. The abstract mathematics, therefore, \emph{explains} the structure of terms and \emph{justifies} the use of ordinary equational reasoning. The elements of Lawvere's approach to universal algebra are illustrated in the left side of Fig.~\ref{fig:hexagons}.

\medskip
% partiality - examples.
% Equational reasoning in Computer science
In this paper we are concerned with \emph{partial} algebraic structures, i.e. those where the operations are not, in general, defined on their whole domain. Partiality is important in mathematics: the very notion of \emph{category} itself is a partial algebraic structure, since only compatible pairs of arrows can be composed. Even more so, partiality is an essential property of computation, and partial functions play a role in many different parts of computer science, starting with initial forays into recursion theory at the birth of the subject, and being ever present in more recent developments, for example arising as an essential ingredient in the study of fixpoints~\cite{bloom1993iteration}. From the start it was clear that additional care is necessary for partial operations, the terms built up from them, and the associated principles of (partial) equational reasoning. An example is the principle of \emph{Kleene equality}: using $s=t$ to assert that whenever one side is defined, so is the other, and they are equal, or the use of notation $-\!\!\mid_X$ to restrict the domain of definition of a function. In general, reasoning about partially defined terms can be quite subtle.

%syntactic approaches and their failures
% essential algebraic theories, finite product sketches are a kind of syntax, but they don't come with a usable notion of equational reasoning.
% semantics correspondences, Gabriel-Ulmer

\medskip
\textbf{Our contribution} is summarised on the right hand side of Fig.~\ref{fig:hexagons} and follows Lawvere's approach closely.
A key question that we address is what replaces the central notion of cartesian category. This turns out to be the notion of
\emph{discrete cartesian restriction} (DCR) category, which arose from research on the structure of partiality~\cite{Coc12}. Just as the free cartesian category on one object plays a central role in the definition of Lawvere theory, the free DCR category on one object plays a central role in the definition of \emph{partial} Lawvere theory that we propose. In our development, the category of sets and partial functions $\Par$ replaces $\Set$ as the universe of models. Much of the richness of the classical picture is unchanged: e.g.\ we obtain free models just as in the classical setting. Moreover, we prove a variety theorem that characterises partial varieties as locally finitely presentable (LFP) categories.

\medskip
% string diagrams
Props~\cite{maclane1965categorical,compprops} and their associated \emph{string diagrams}, play a crucial technical role. Props are a convenient categorical structure that capture generic \emph{monoidal} theories. Monoidal theories differ from equational theories in that, roughly speaking, \fix{in that we are able to consider more general monoidal structures other than the} cartesian \fix{one}. String diagrams are the syntax of props, and they are a bona fide syntax not far removed from traditional terms. For example, they can be recursively defined and enjoy similar properties as free objects, e.g.\ the principle of structural induction. The connective tissue between the classical story and string diagrams is Fox's Theorem~\cite{Fox76}, which states that the structure of cartesian categories can be captured by the presence of local algebraic structure: a coherent and natural commutative comonoid structure on each object. This implies several things: (\textit{i}) that classical terms can be seen as particular kinds of string diagrams, (\textit{ii}) that classical equational reasoning can be seen as diagrammatic reasoning on these string diagrams and (\textit{iii}) that the prop induced from the monoidal theory of commutative comonoids --- well-known to coincide with $\bbF^\op$, the opposite of the prop of finite sets and functions --- is the free cartesian category on one object. The correspondence goes further: as shown in~\cite{Bonchi2018},
Lawvere theories are particular kinds of monoidal theories.

We are able to identify the nature of the free DCR category on one object by proving a result similar to Fox's Theorem, but for DCR categories instead of cartesian categories. Instead of commutative comonoids, we identify the algebraic structure of interest as \emph{partial Frobenius algebras}. The free DCR category on one object is the prop induced from this monoidal theory, and it can be characterised as $\parcat{\F^\op}$: the prop of partial functions in $\F^\op$. This informs our definition of partial Lawvere theory. Crucially, just as the mathematics of ordinary Lawvere theories serves as a post hoc justification for equational theories, we identify the precise class of string diagrams that serve as \emph{partial terms}, which lets us define a \emph{partial equational theory} in a familiar way as pair of signature and equations. We give several examples, from partial commutative monoids, to several examples important in computer science, notably the theory of partial combinatory algebras~\cite{Bethke88}, the theory of pairing functions, and the theory of cartesian closed categories.

To summarise, the original contributions of this paper are:
\begin{itemize}
	\item A ``Fox theorem'' for DCR categories, which uses the notion of \emph{partial Frobenius algebra} and leads to the characterisation of the free DCR category on one object as $\parcat{\F^\op}$;
	\item The definitions of partial Lawvere theory and partial equational theory, which use string diagrammatic syntax informed by the aforementioned Fox theorem;
	\item The coupling of these notions into a comprehensive framework for partial algebraic theories, analogous to the work of Lawvere on classical algebraic theories, as illustrated in Fig.~\ref{fig:hexagons};
	\item The existence of free models, and---more generally--a variety theorem, building on known results about DCR categories, and the Gabriel-Ulmer duality.
\end{itemize}

\paragraph{Related work.}
There are a number of formalisms in the literature that aim at providing a rigorous way of specifying partial algebraic structure. Freyd's \emph{essentially algebraic theories}~\cite{freyd1972aspects} were introduced informally, but were subsequently formalised and generalised in various ways~\cite{adamek_rosicky_1994, palmgren_vickers_2007, Adamek2011}.
A different, but equally expressive approach is via \emph{finite limit sketches}~\cite{adamek_rosicky_1994}. Nevertheless, none of these approaches can claim to have the foundational status of classical equational theories - e.g.\ they do not, per se, provide a canonical notion of syntax to replace classical terms, nor a calculus for (partial) equational reasoning about the categories of models they define. Tout court, none of them can claim to be \emph{equational}. Interestingly, the semantic landscape (i.e.\ the corresponding notion of \emph{partial variety}) is better understood than the syntax. The class of models of essentially algebraic theories and finite limit sketches are closely related to Gabriel-Ulmer duality~\cite{centazzo2004generalised}, which asserts a contravariant (bi)equivalence between the category of categories with finite limits and the category of locally finitely presentable (LFP) categories.

Partial Frobenius algebras, which arise in our characterisation of DCR categories, are special/separable Frobenius algebras without units. The version \emph{with} units was originally studied in~\cite{Carboni1987}, is deeply connected to the relational algebra~\cite{Freyd1990}, characterises 2-dimensional TQFTs~\cite{Kock2003}, and has been used extensively in categorical approaches to the study of quantum information and quantum computing, such as the ZX calculus~\cite{Coecke2008}. In a similar way to the use of partial Frobenius algebras in this paper, they are used in the recently proposed \emph{Frobenius theories}~\cite{Bonchi2017c}, which are algebraic theories that take their models in the category of relations $\ct{Rel}$, and are guided by the structure of cartesian bicategories of relations~\cite{Carboni1987}.

Restriction categories were introduced in~\cite{Coc02} as a general framework for reasoning about categories of partial maps. Cartesian restriction (CR) categories are those with a certain sort of formal finite product structure -- restriction products -- introduced in~\cite{Coc07}. Notably, the p-categories of~\cite{Rob88} arise as restriction categories with restriction products. Discrete cartesian restriction (DCR) categories -- named for a similarity to categories of discrete topological spaces -- arise in \cite{Coc12} as the restriction categories with finite \emph{latent} limits -- again a sort of formal limit. DCR categories are closely connected to the \emph{discrete inverse categories} considered in \cite{Gil14} which are presentable in terms of \emph{semi-Frobenius} algebras, being those special/separable commutative Frobenius algebras with neither a unit \emph{nor a counit}.

\paragraph{Structure of the paper.} In \S\ref{sec:history} we lay the foundations by recalling the basic concepts of universal algebra, props and string diagrams, Fox's theorem, and functorial semantics. In \S\ref{sec:partialmaps}, after recalling the basics of restriction category theory, we prove Theorem~\ref{thm:newfox}, which is to DCR categories what Fox's theorem is to cartesian categories. In \S\ref{parlawve} we propose our original definitions: partial Lawvere theories and their varieties.  Next, \S\ref{sec:partialequationaltheories} is devoted to the associated notion of partial equational theory, and several examples, continued in \S\ref{sec:multi} with multi-sorted examples. Our variety theorem is in \S\ref{the_variety} where we also treat other semantic aspects, e.g.\ the existence of free models.

\section{Background material}\label{sec:history}

\subsection{Overview of classical universal algebra}
Universal algebra is the study of \textit{equational theories} and of their semantics, \emph{varieties}. In this section we recall the basic concepts and definitions.
\begin{definition}\label{unialg}
	A \emph{signature} is a pair $(\Sigma,\alpha)$ where $\Sigma$ is a set
	and $\alpha$ a function $\Sigma \to \bbN$ that assigns to every element $t : \Sigma$ a natural number $\alpha(n): \bbN$ called the \emph{arity} of the \emph{function symbol} $t$.
\end{definition}
\begin{notation}\rm
	The arity ``slices'' the set $\Sigma$ of function symbols.	The slice $\Sigma_n \subseteq\Sigma$ contains operations of arity $n$, and $t : \Sigma_n$ is a synonym for ``$t$ is a $n$-ary operation''. We will sometimes write $t_n$ for a generic element of $\Sigma_n$.
	We shall refer to the signature as just $\Sigma$ if the arity function is understood from the context. For example the signature $\Sigma_{\ct{M}}$ of \emph{monoids} is $\{m,\,e\}$, with $\alpha(m)=2$ and $\alpha(e)=0$.
\end{notation}
\begin{definition}\label{defn:sigmaalgebra}
	A $\Sigma$-\emph{algebra} is a pair $(A,\semfun{-}_A)$ where $A$ is a set and $\semfun{-}_A$ is a function sending  function symbols $t : \Sigma_n$ to functions $\semfun{t}_A\from A^n \ra A$. The function $\semfun{t}_A$ is called the $n$-ary operation on $A$ \emph{associated} to the function symbol $t: \Sigma_n$. We refer to $A$ as the \emph{carrier} of the $\Sigma$-algebra.
	A $\Sigma$-algebra homomorphism from $(A,\semfun{-}_A)$ to $(B,\semfun{-}_B)$ is a function $f\from A\to B$ that respects the $\Sigma$ structure: i.e.\ for every $n\in\bbN$ and $t : \Sigma_n$, the following diagram commutes:
	\[
		\begin{tikzcd}
			{A^n} \ar[d,"\semfun{t}_A"'] \ar[r, "f^n"] & {B^n} \ar[d,"\semfun{t}_B"] \\
			A \ar[r, "f"] & B.
		\end{tikzcd}
	\]
\end{definition}
\begin{remark}
	$\Sigma$-algebras and their homomorphisms define a category $\clV_\Sigma$.
\end{remark}
Of course, an algebraic structure isn't just about operations, but also about \emph{properties} enjoyed by those operations. To express this we first need the notion of \emph{term}. Fixing a signature $\Sigma$, we recall the usual recursive construction of the set of terms $T_\Sigma^{V}$, for some set of variables $V$:
\[
	T_\Sigma^V \bnfEq V \bnfSep t_0 \bnfSep t_1(T_\Sigma^V) \bnfSep t_2(T_\Sigma^V,T_\Sigma^V) \bnfSep \dots \bnfSep t_n(T_\Sigma^V,\dots,T_\Sigma^V) \bnfSep \dots
\]
In the above, each $t_i$ ranges over the function symbols in $\Sigma_i$.
For any $V$, $T^V_\Sigma$ carries a canonical $\Sigma$-algebra structure:
$\semfun{t}(t_1,t_2,\dots,t_{n_t}) \Defeq t(t_1,t_2,\dots,t_{n_t})$. We call this the term $\Sigma$-algebra over $V$.

\begin{observation}
	The term $\Sigma$-algebra $T_\Sigma^V$ enjoys a universal property: given
	a $\Sigma$-algebra $(A,\semfun{-}_A)$ and function $v \from V\to A$, there is a unique extension to a homomorphism of algebras $\bar v \from T_{\Sigma}^V \to A$. This is just the induction principle associated to the recursive definition of terms.
\end{observation}

\begin{definition}[$\Sigma$-equation]
	Fixing $V$,
	a $\Sigma$-\emph{equation} is a pair $(s,t) \in T^V_\Sigma\times T^V_\Sigma$;
	we usually write `$s=t$'. % for a $\Sigma$-equation $(s,t)$.
	A $\Sigma$-equation $s=t$ \emph{holds} in $\Sigma$-algebra $(A,\semfun{-}_A)$ if
	for all $v\from V\to A$ we have $\bar v (s)= \bar v (t)$ in $A$.
\end{definition}
Given the signature of monoids, we can express properties
such as associativity: $m(x,m(y,z)) = m(m(x,y),z)$; or commutativity: $m(x,y)=m(y,x)$; etc. The idea is that a set of $\Sigma$-equations \emph{constrains} the choice of algebras $(A,\semfun{-}_A)$ to those where every equation holds.
\begin{definition}[Equational Theory and Variety]\label{def:model}
	A pair $(\Sigma, E)$ where $\Sigma$ is a signature and $E$ a set of $\Sigma$-equations is called an \emph{equational theory}. A \emph{model} of $(\Sigma,E)$ is a $\Sigma$-algebra where every $e:E$ holds.
	The class of models for an equational theory is called a \emph{variety}.
\end{definition}

\begin{example}\label{ex:eqtheorymonoids}
	The equational theory of commutative monoids is
	\[
		(\,\{m,\,e\},\,\{\, m(m(x,y),z) = m(x,m(y,z)),\, m(x,y)=m(y,x),\, m(e,x)=x\,\}\,).
	\]
	The corresponding variety is the class of commutative monoids.
\end{example}
%\begin{example}\label{ex:varieties1} We list some examples of varieties:
%	\begin{itemize}
%		\item The class of sets and functions.
%		%\item The class monoids and monoid homomorphisms.
%		\item The class commutative monoids and monoid homomorphisms.
%		\item The class of groups and group homomorphisms.
%	\end{itemize}
%\end{example}
Some of the most famous results of universal algebra characterise varieties. For example:
%In particular, we recall the celebrated Birkhoff variety theorem.
\begin{theorem}[Birkhoff~\cite{birkhoff1935structure}]\label{birko}
	A class of $\Sigma$-algebras is a variety if and only if it is closed under homomorphic images, subalgebras and products.
\end{theorem}
\subsection{Props and monoidal theories}
%This paper focuses on the passage from ordinary algebraic theories to \emph{partial} algebraic theories.
Our development is informed by the differences between the algebraic structure of total functions and partial functions. Given the focus on algebra, the notion of prop is useful as a categorical gadget on which to hang an algebraic structure. Moreover, the associated notion of string diagram will lead us to a syntax with which to express partial equational theories by appropriately generalising classical terms. Here we recall the basic definitions of props \cite{compprops}, string diagrams and some of the algebraic structures that are prominent in subsequent sections.
\begin{definition}[Prop \protect{\cite[Ch. 5]{maclane1965categorical}}]
	A \emph{prop} is a symmetric strict monoidal category with set of objects the natural numbers $\N$, where the monoidal product on objects is addition: $m\otimes n := m+n$. A \emph{homomorphism of props} is an identity-on-objects symmetric strict monoidal functor.
\end{definition}

\begin{example}\label{ex:F}
	An important example is the prop $\bbF$ of finite ordinal numbers. In the following, $[m]\Defeq\{1,2,\dots,m\}$.
	The $\bbF$-arrows $m\to n$ are all functions $[m]\to [n]$: composition is function composition, and the monoidal product is ``disjoint union''; i.e.\ for $f_1\from m_1\to n_1$ and $f_2\from m_2\to n_2$,
	\[
		(f_1\otimes f_2)(i)\from m_1+m_2\to n_1+n_2 \Defeq
		\begin{cases}
			f_1(i) & \text{if }i\leq m_1 \\ f_2(i-m_1)+n_1 & \text{otherwise.}
		\end{cases}\]
\end{example}
%$\F$ is the skeleton of the category of finite sets and functions $\FinSet$, indeed there is an equivalence of categories $\F \simeq \FinSet$.

Free props generated from some signature of operations are of particular importance.
\begin{definition}[Monoidal signature]\label{monosigna}
	A monoidal signature $\Gamma$ is a collection of \emph{generators} $\gamma : \Gamma$, each with an arity $\text{ar}(\gamma) : \N$ and coarity $\text{coar}(\gamma) : \N$.
\end{definition}

Concrete terms can be given a BNF description, as follows:
\begin{equation}
	\label{eq:syntax}%\tag{BNFT}
	c \bnfEq \gamma\in \Gamma \bnfSep \,\raisebox{.25em}{\begin{tikzpicture}[scale=.5,baseline=(current bounding box.center)]
\akasa
\end{tikzpicture}} \bnfSep \raisebox{.25em}{\begin{tikzpicture}[scale=.5,baseline=(current bounding box.center)]
\draw[iden] (0,2*\len) -- (1,2*\len);
\end{tikzpicture}} \bnfSep \!\!\symDiag \bnfSep c\otimes c \bnfSep c \comp c
\end{equation}
Arities and coarities are not handled in the BNF but with an associated sorting discipline, shown below. We only consider terms that have a sort, which is unique if it exists.
\begin{gather*}
	\reductionRule{}{ \typeJudgment{\gamma : \Gamma}{\gamma}{\sort{\text{ar}(\gamma)}{\text{coar}(\gamma)}} }\qquad
	\reductionRule{}{ \typeJudgment{}{\raisebox{.25em}{\begin{tikzpicture}[scale=.5,baseline=(current bounding box.center)]
\akasa
\end{tikzpicture}}}{\sort{0}{0}} }\qquad
	\reductionRule{}{ \typeJudgment{}{\raisebox{.25em}{\begin{tikzpicture}[scale=.5,baseline=(current bounding box.center)]
\draw[iden] (0,2*\len) -- (1,2*\len);
\end{tikzpicture}}}{\sort{1}{1}} }\qquad
	\reductionRule{}{ \typeJudgment{}{\symDiag}{\sort{2}{2}} }\\[1em] 
	%===
	\reductionRule{ \typeJudgment{}{c}{\sort{n}{z}} \quad \typeJudgment{}{d}{\sort{z}{m}} }
	{ \typeJudgment{}{c\comp d}{\sort{n}{m}} }\qquad
	\reductionRule{ \typeJudgment{}{c}{\sort{n}{m}} \quad \typeJudgment{}{d}{\sort{r}{z}} }
	{ \typeJudgment{}{c \otimes d}{\sort{n+r}{m+z}} }
\end{gather*}
The idea is that the sort $c : \sort{m}{n}$ counts the number of ``dangling wires'' of each term. Every sortable term generated from~\eqref{eq:syntax} has a diagrammatic representation.
The convention for $\gamma : \Gamma$ is to draw it as a box with $\text{ar}(\gamma)$ ``dangling wires'' on the left and $\text{coar}(\gamma)$ on the right:
\[
	% \text{FIX }\sigma \ra \gamma : \qquad	\graffle[5pt]{.6cm}{sigma}.
	\text{ar}(\gamma)
	\Big\{
	\begin{tikzpicture}[baseline=(current bounding box.center)]
		\id\up[.5]\id\down\id
		\draw[fill=white] (\len,\len-.1) rectangle (3*\len,3*\len) node[pos=.5]
			{$\gamma$};
		\up[1.625]{\node {\scalebox{.4}{$\vdots$}};}
		\step{\up[1.625]{\node {\scalebox{.4}{$\vdots$}};}}
		% \up[2]{\node {\tiny $\vdots$};}
	\end{tikzpicture}
	\Big\}\text{coar}(\gamma)
\]
The conventions for the  \eqref{eq:syntax} operations are:
$
	c \comp c' \text{ is drawn}
	\graffle[9pt]{.8cm}{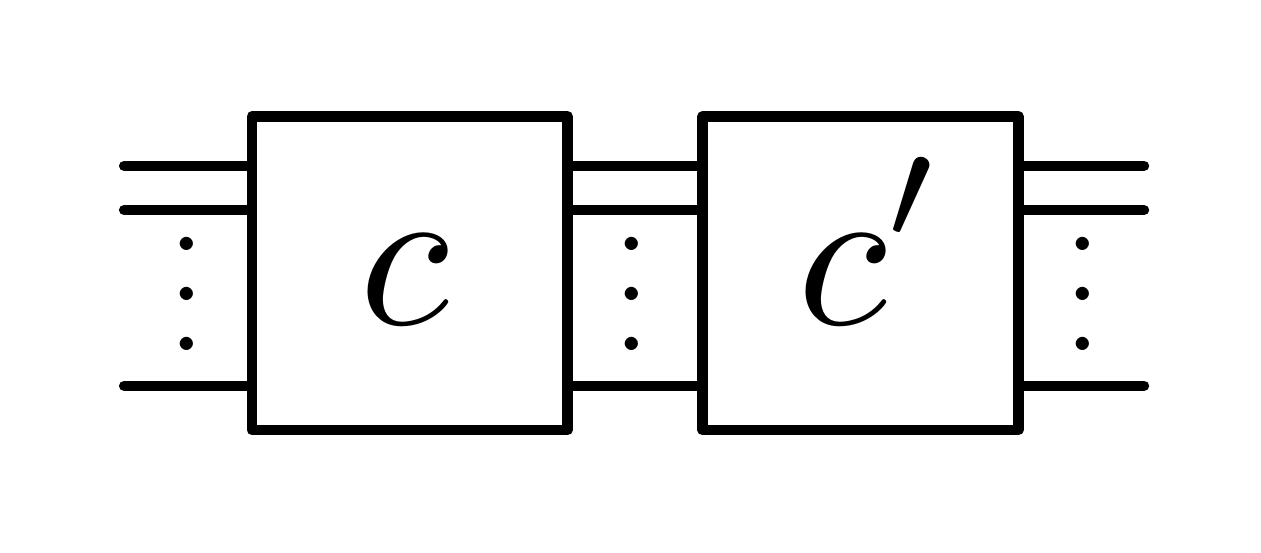}
	\;
	\text{and } c \otimes c' \text{ is drawn}
	\graffle[15pt]{1.2cm}{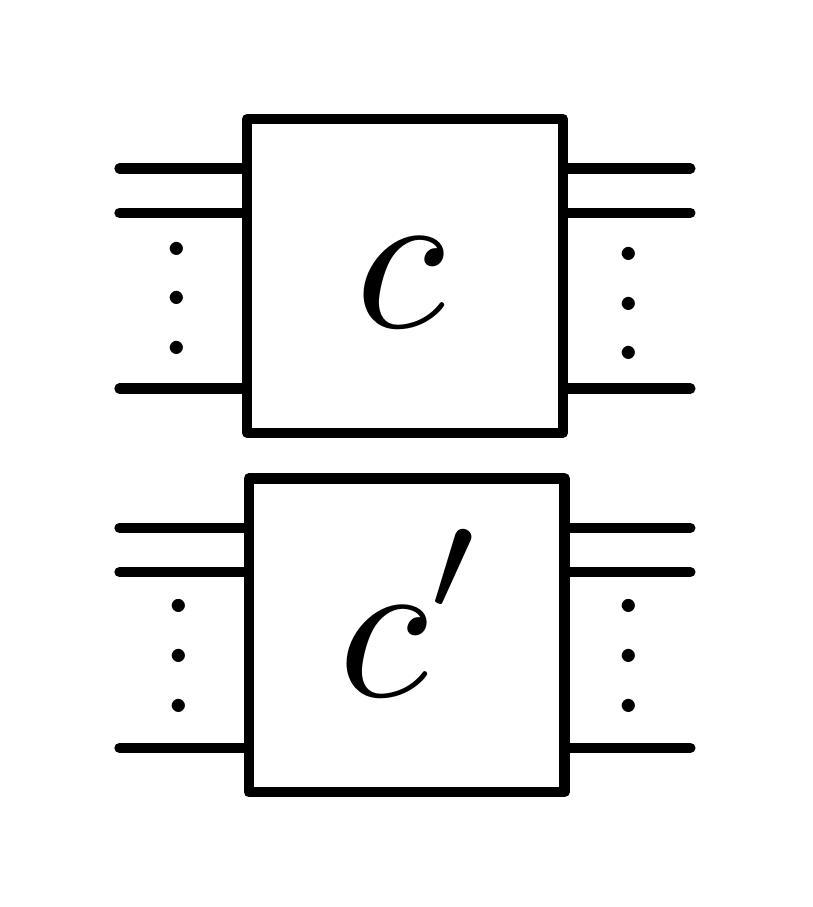}.
$
The sorting discipline ensures that the convention for $\comp$ makes sense.

\begin{example}\label{ex:generators}
	Consider the following signature, where the (co)arities are
	apparent from the
	\begin{equation}\label{eq:monoidgenerators} \tag{CMG}
		\Gamma \Defeq \Big\{ \raisebox{.25em}{\begin{tikzpicture}[scale=.75,baseline=(current bounding box.center)]
\mult
\end{tikzpicture}},\ \raisebox{.25em}{\begin{tikzpicture}[xscale=.5,baseline=(current bounding box.center)]
\unit
\end{tikzpicture}} \Big\}
	\end{equation}
	glyphs. The term
	$(\raisebox{.25em}{\begin{tikzpicture}[scale=.75,baseline=(current bounding box.center)]
\mult
\end{tikzpicture}} \otimes (\raisebox{.25em}{\begin{tikzpicture}[scale=.5,baseline=(current bounding box.center)]
\draw[iden] (0,2*\len) -- (1,2*\len);
\end{tikzpicture}} \otimes \raisebox{.25em}{\begin{tikzpicture}[xscale=.5,baseline=(current bounding box.center)]
\unit
\end{tikzpicture}})) \comp
		((\raisebox{.25em}{\begin{tikzpicture}[scale=.75,baseline=(current bounding box.center)]
\mult
\end{tikzpicture}} \otimes \raisebox{.25em}{\begin{tikzpicture}[scale=.5,baseline=(current bounding box.center)]
\draw[iden] (0,2*\len) -- (1,2*\len);
\end{tikzpicture}}) \comp \symDiag)$
	has sort $\sort{3}{2}$ and diagram:
	\[
		\graffle{1.5cm}{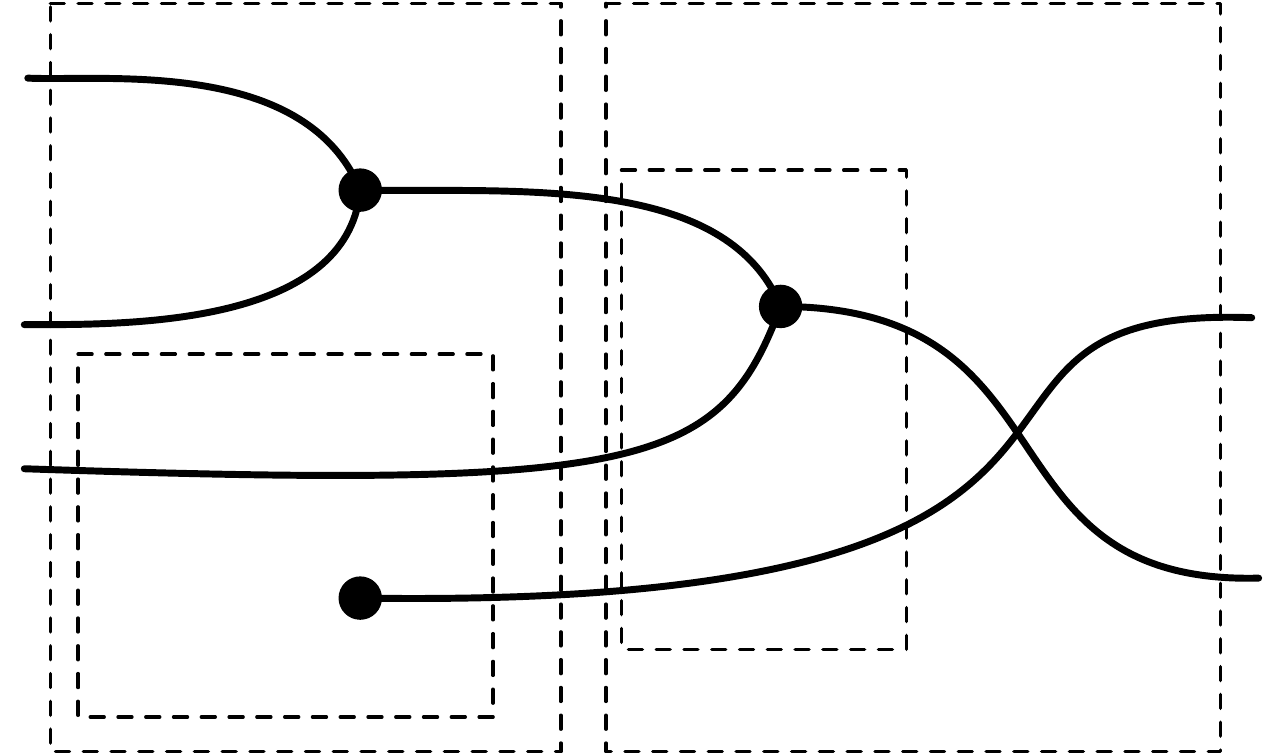}
	\]
	where the ``dotted line'' boxes serve the role of parentheses.
\end{example}

Terms of~\eqref{eq:syntax} are quotiented by the laws of symmetric strict monoidal categories. We do not go into the details here, but these are closely connected with the diagrammatic conventions. Indeed, they allow us to discard the ``dotted line'' boxes and focus \emph{only} on the connectivity between the generators. For example the following two diagrams are in the same equivalence class of terms:
\[
	\graffle[13pt]{1.3cm}{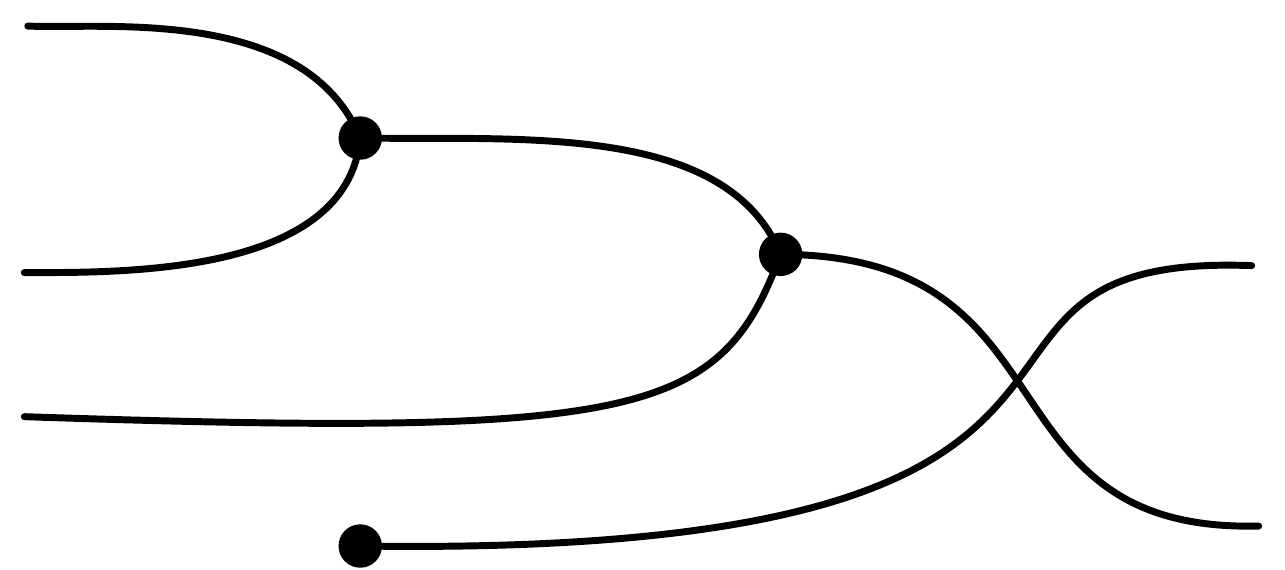}
	\quad=\quad
	\graffle[10pt]{1cm}{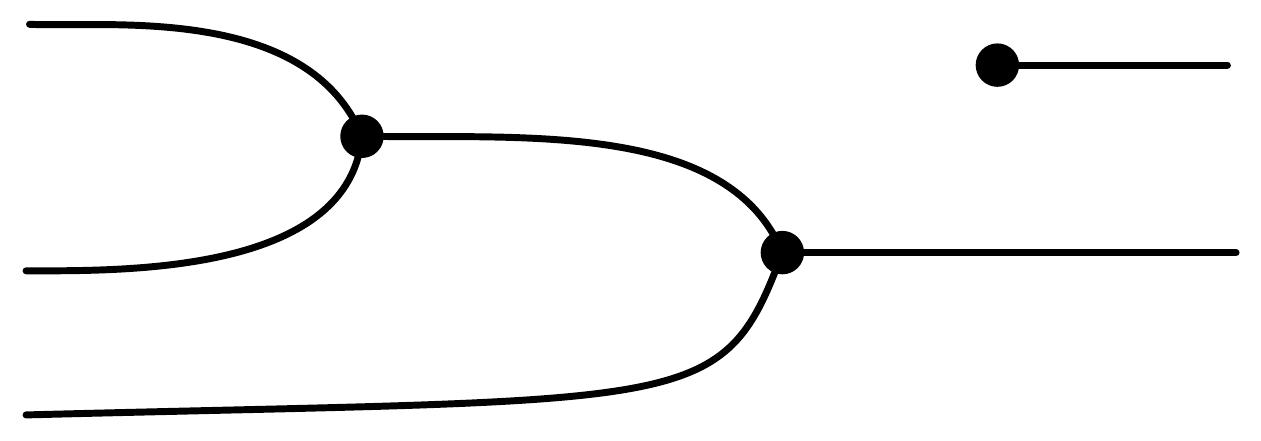}
\]
We refer to equivalence classes $[c]\from m\to n$ as \emph{string diagrams}.
\begin{definition}
	The free prop $\bbX_\Gamma$ on $\Gamma$ has as arrows $m\to n$ string diagrams $[c]:\sort{m}{n}$.
\end{definition}

String diagrams can be used to specify additional equations that specify algebraic structure.
\begin{definition}[Monoidal theory \cite{compprops}]
	For a monoidal signature $\Gamma$, a $\Gamma$-equation is a pair $([c],[d])$ of equally-sorted string diagrams; we usually write  `$[c]=[d]$'. % for a $\Gamma$-equation $([c],[d])$.
	A \emph{monoidal theory} is a pair $(\Gamma,F)$ where $F$ is a set of $\Gamma$-equations.
\end{definition}
Given a monoidal theory $(\Gamma,F)$, the induced prop $\bbX_{(\Gamma,F)}$ can be obtained by taking a coequaliser in $\Cat$. It can alternatively be given an explicit description as follows: as arrows $[m]\to [n]$ it has arrows of $\bbX_\Gamma$ quotiented by the smallest congruence containing $F$.

\begin{example}\label{ex:propmonoid}
	Consider the signature~\eqref{eq:monoidgenerators} and the following set of equations:
	\begin{equation}\label{eq:monoid}\tag{CM}
		E \Defeq
		\Big\{\hspace{.5em}
		\raisebox{.5em}{\begin{tikzpicture}[xscale=.75,baseline=(current bounding box.center)]
				\umult
				\lid
				\step{\mult}
				\ezs{\lmult
					\uid
					\step{\mult}}{xshift=3cm,yshift=\len cm}
				\node at (2.5,2*\len) {$=$};
			\end{tikzpicture}}\;,\hspace{.5em}
		\raisebox{.5em}{\begin{tikzpicture}[xscale=.75,baseline=(current bounding box.center)]
				\braid\step{\mult}\ezs{\mult}{xshift=3cm}
				\node at (2.5,2*\len) {$=$};
			\end{tikzpicture}}\;,\hspace{.5em}
		\raisebox{.5em}{\begin{tikzpicture}[xscale=.75,baseline=(current bounding box.center)]
				\ezs{\unit}{yshift=\len cm}\lid
				\ezs{\id}{xshift=3cm}
				\step{\mult}
				\node at (2.5,2*\len) {$=$};
			\end{tikzpicture}}
		\hspace{.5em}\Big\}.
	\end{equation}
	The resulting prop $\bbC\bbM$ is the prop of \emph{commutative monoids}. The equations, from left to right, express associativity, commutativity and unitality.
	\begin{remark}\label{reasoning_w_diag}
		String diagrams in $\bbX_{(\Gamma,F)}$ are amenable to equational reasoning, often referred to as \emph{diagrammatic reasoning} in this context: if $([c],[d]) \in F$ then substituting $c$ for $d$ inside any context is sound. For example in $\bbC\bbM$ the set of equations contains only one of the unit laws. The other may be derived:
		\begin{gather*}
\begin{tikzpicture}[xscale=.75,baseline=(current bounding box.center)]
\down\unit\uid
\step{\mult}
\ezs{\uid\down\unit
\step{\mult}
\step\akasa
}{xshift=3cm}
\ezs{\uid\down\unit
\step\braid
\ezs{\step[.5]\akasa}{xscale=2}
\step[2]\mult
}{xshift=6cm}
\ezs{
\uid\down\unit
\step\braid
\step[2]\mult
}{xshift=10cm}
\eql{2.5}
\eql{5.5}
\eql{9.5}
\end{tikzpicture}\\[-3mm]
\begin{tikzpicture}[xscale=.75,baseline=(current bounding box.center)]
\up\unit\lid
\step{\mult}
\ezs{\up\unit\lid
\step{\mult}
\ezs{\step[.25]\akasa}{xscale=1.25}
}{xshift=3cm}
\ezs{
\id
\step[.25]{\ezs{\akasa}{xscale=.5}}
}{xshift=6cm}
\ezs{
\id}{xshift=8cm}
\eql{-.5}
\eql{2.5}
\eql{5.5}
\eql{7.5}
\end{tikzpicture}
\end{gather*}
		We typically omit the ``dotted line'' boxes in such chains of reasoning.
	\end{remark}
\end{example}
Interestingly, $\bbC\bbM$ can be seen as the algebraic characterisation of $\F$.
\begin{observation}[\cite{compprops}]\label{obs:Fmonoids}
	As props, $\F\cong \bbC\bbM$.
\end{observation}
\begin{remark}\label{rem:pictures}
	In fact, arrows of $\bbC\bbM$ can be intuitively understood as ``pictures of functions''.
	For example, the function $f\from 2\ra 2$ where $f(1)=f(2)=1$ is drawn\ \begin{tikzpicture}[baseline=(current bounding box.center)].
		\up[2]\mult
		\unit
	\end{tikzpicture}
\end{remark}
\begin{example}\label{ex:comonoids}
	The theory of commutative \emph{comonoids} plays an important role for us. The data is:
	\begin{equation}
		\label{eq:comonoidgenerators} \tag{CCMG}
		\raisebox{.5em}{\begin{tikzpicture}[baseline=(current bounding box.center)]
				\comult
			\end{tikzpicture}\hspace{2em}
			\begin{tikzpicture}[baseline=(current bounding box.center)]
				\counit
			\end{tikzpicture}}
	\end{equation}
	\vspace{-.5cm}
	%and the equations are:
	\begin{equation}\label{eq:comonoid}\tag{CCM}
		\raisebox{.5em}{\begin{tikzpicture}[xscale=-.75,baseline=(current bounding box.center)]
				\umult
				\lid
				\step{\mult}
				\ezs{\lmult
					\uid
					\step{\mult}}{xshift=3cm,yshift=\len cm}
				\node at (2.5,2*\len) {$=$};
			\end{tikzpicture}\hspace{2em}
			\begin{tikzpicture}[xscale=-.75,baseline=(current bounding box.center)]
				\braid\step{\mult}\ezs{\mult}{xshift=3cm}
				\node at (2.5,2*\len) {$=$};
			\end{tikzpicture}\hspace{2em}
			\begin{tikzpicture}[xscale=-.75,baseline=(current bounding box.center)]
				\ezs{\unit}{yshift=\len cm}\lid
				\ezs{\id}{xshift=3cm}
				\step{\mult}
				\node at (2.5,2*\len) {$=$};
			\end{tikzpicture}}
	\end{equation}
	Let $\bbC\bbC$ be the prop induced from the monoidal theory (\eqref{eq:comonoidgenerators},\eqref{eq:comonoid}).
\end{example}
Given that~\eqref{eq:comonoidgenerators} and~\eqref{eq:comonoid} are mirrored~
\eqref{eq:monoidgenerators} and~\eqref{eq:monoid}, Observation~\ref{obs:Fmonoids} gives:
\begin{observation}\label{obs:Fopcomonoids}
	As props, $\F^\op \cong \bbC\bbC$.
\end{observation}
While we have specialised our discussion of string diagrams as the syntax of props, it is well-known that they can be used
as a sound calculus in any symmetric (strict) monoidal category. Roughly speaking, objects are represented by wires, and morphisms by boxes.
%We end our discussion of monoidal theories by observing that the machinery above can be adapted to reason about morphisms in arbitrary
%symmetric strict monoidal categories: the arities and coarities become lists of generating objects, and the sorting discipline changes to match.
\subsection{Fox's theorem}
%\fos{Pawel introduces this section with a crisp explaining sentence about its relationship with the previous?}
Equational and monoidal theories are linked by Fox's theorem (\cite{Fox76}), recalled here -- this will be explained in~\S\ref{subsec:termsasdiagrams}.
%It is possible to capture important universal properties in terms of monoidal theories. Here we discuss the case of
%
%Recall that
\emph{Cartesian} categories are categories with finite products, and \emph{cartesian functors} preserve them. Fox showed %(Theorem~\ref{thm:fox})
that cartesian categories are exactly those that have a certain algebraic structure.

A commutative comonoid on an object $X$ of a symmetric monoidal category $\bbX$ is a triple $(X,\delta_X,\varepsilon_X)$ s.t.\  $\delta_X : X \to X \otimes X$ and $\varepsilon_X : X \to I$, depicted as $\raisebox{.25em}{\begin{tikzpicture}[scale=.75,baseline=(current bounding box.center)]
\comult
\end{tikzpicture}}$ and $\raisebox{.25em}{\begin{tikzpicture}[xscale=.5,baseline=(current bounding box.center)]
\counit
\end{tikzpicture}}$ respectively, and these satisfy (CCM). %Now if $\bbX$ is a symmetric monoidal category in which
If all objects %$X$ of $\bbX$
are so equipped,  %with such a commutative comonoid structure,
%$(X,\delta_X,\varepsilon_X)$,
then the structures are \emph{coherent} if for all objects $X,Y$:
\begin{equation}\label{eq:coherent}\tag{coherent}
	\begin{tikzpicture}[xscale=.75,baseline=(current bounding box.center)]
		\comult
		\node[left] at (0,2*\len) {$\scriptscriptstyle X\otimes Y$};
		\node[right] at (1,\len) {$\scriptscriptstyle X\otimes Y$};
		\node[right] at (1,3*\len) {$\scriptscriptstyle X\otimes Y$};
		\ezs{
			\comult
			\ezs{\comult}{yshift=1cm}
			\ezs{\braid}{xshift=1cm, yshift=.5cm}
			\ezs{\lid}{xshift=1cm}
			\ezs{\uid}{xshift=1cm, yshift=1cm}
			\node[left] at (0,2*\len) {$\scriptscriptstyle X$};
			\node[left] at (0,6*\len) {$\scriptscriptstyle Y$};
			\node[right] at (2,\len) {$\scriptscriptstyle X$};
			\node[right] at (2,5*\len) {$\scriptscriptstyle X$};
			\node[right] at (2,3*\len) {$\scriptscriptstyle Y$};
			\node[right] at (2,7*\len) {$\scriptscriptstyle Y$};
		}{xshift=3cm,yshift=-.5cm}
		\node at (2.25,2*\len) {\Scale[1.25]{=}};
	\end{tikzpicture}
	\qquad
	\begin{tikzpicture}[baseline=(current bounding box.center)]
		\counit
		\node[left] at (0,2*\len) {$\scriptscriptstyle X\otimes Y$};
		\ezs{
			\ezs{\counit}{yshift=-\len cm}
			\ezs{\counit}{yshift= \len cm}
			\node[left] at (0,\len) {$\scriptscriptstyle X$};
			\node[left] at (0,3*\len) {$\scriptscriptstyle Y$};
		}{xshift=2cm}
		\node at (1.225,2*\len) {\Scale[1.25]{=}};
	\end{tikzpicture}
\end{equation}
Further, we say that the $\delta$ and $\varepsilon$ are \emph{natural} if for any arrow $f : X \to Y$ of $\bbX$, we have:
\begin{equation}\label{eq:natural}\tag{natural}
	\begin{tikzpicture}[baseline=(current bounding box.center),xscale=.5]
		\mor{f}
		\ezs{\comult}{xshift=1cm}
		\node[left] at (0,2*\len) {$\scriptscriptstyle X$};
		\node[right] at (2,\len) {$\scriptscriptstyle Y$};
		\node[right] at (2,3*\len) {$\scriptscriptstyle Y$};
		\ezs{
			\comult
			\ezs{\umor{f}}{xshift=1cm}
			\ezs{\lmor{f}}{xshift=1cm}
			\node[left] at (0,2*\len) {$\scriptscriptstyle X$};
			\node[right] at (2,\len) {$\scriptscriptstyle Y$};
			\node[right] at (2,3*\len) {$\scriptscriptstyle Y$};
		}{xshift=5cm}
		\node at (3.5,2*\len)  {\Scale[1.25]{=}};
	\end{tikzpicture}
	\hspace{1cm}
	\begin{tikzpicture}[baseline=(current bounding box.center),xscale=.5]
		\mor{f}
		\ezs{\counit}{xshift=1cm}
		\ezs{\counit}{xshift=4cm,xscale=2}
		\node[left] at (0,2*\len) {$\scriptscriptstyle X$}; %here
		\node at (2.75,2*\len) {$=$};
	\end{tikzpicture}
\end{equation}

\begin{theorem}[\protect{\cite{Fox76}}]\label{thm:fox}
  A cartesian category is the same thing as a symmetric monoidal category
	where every object is equipped with a \eqref{eq:coherent} and \eqref{eq:natural} commutative comonoid structure.
	%A symmetric monoidal category is cartesian if and only if
	%every object can be equipped with a commutative comonoid structure that is \eqref{eq:coherent} and \eqref{eq:natural}.
\end{theorem}

In light of Observation~\ref{obs:Fopcomonoids}, we know that a commutative comonoid structure on $X$ is equivalently a cartesian functor $\clX:\bbF^{\op}\to \bbX$ where $\clX[1]=X$. The action of $\clX$ on objects is determined by its action on 1, and the generators give arrows $\clX(\raisebox{.25em}{\begin{tikzpicture}[scale=.75,baseline=(current bounding box.center)]
\comult
\end{tikzpicture}}) = \delta_X \from X \to X\otimes X$ and $\clX({\scriptstyle X}\,\raisebox{.25em}{\begin{tikzpicture}[xscale=.5,baseline=(current bounding box.center)]
\counit
\end{tikzpicture}}) = \varepsilon_X \from X \to I$ of $\bbX$ which satisfy (CCM). Thus we may specialize Theorem~\ref{thm:fox}, to props as follows:

\begin{corollary}\label{cor:foxprop}
	A prop $\bbX$ is cartesian (with categorical product the monoidal product) if and only if there is a homomorphism of props $\bbF^\op\to\bbX$ and
	the picked out comonoid structure is \eqref{eq:natural}.
\end{corollary}
It is easy to show that a coherent and natural commutative comonoid structure, if it exists, is unique. An easy consequence of Theorem~\ref{thm:fox} is that a cartesian functor is precisely a symmetric monoidal functor that preserves the comonoid structure.
This, combined with Corollary~\ref{cor:foxprop}, gives:
%Together with the restatement of the theorem for props this gives that:
\begin{proposition}\label{Fop_is_free_on_1}
	The prop $\bbF^\op$ is the free cartesian category on a single object.
\end{proposition}
\subsection{Lawvere theories}\label{subsec:law}
We recall Lawvere's approach~\cite{lawvere1963functorial} of \emph{functorial semantics} of algebraic theories in the rest of the section. %Lawvere addressed two perceived weaknesses of the classical
% account. First, the reliance on \emph{presentations}, i.e. a particular choice $(\Sigma,E)$ of function symbols and equations in an equational theory. This is sub-optimal since the same variety might, in general, have many different presentations. Second, models as `sets with structure' is overly restrictive: algebraic structures exist in diverse parts of mathematics; e.g.\ the notion of topological group cannot be easily captured by classical universal algebra.
Lawvere's approach is centered on the theory of cartesian categories.
\begin{definition}[Lawvere theory]\label{lo_tiori}
	A \emph{Lawvere theory} is a cartesian prop.
	A morphism of Lawvere theories is a cartesian prop homomorphism. Lawvere theories and homomorphisms define the category $\textsf{Law}$.
\end{definition}
Finite products do two jobs: they keep track of arities of operations, and---less obviously---they ensure the \emph{totality} and \emph{single-valuedness} of the interpretation of function symbols in any model.

\emph{Free} categories with products play a leading role. %, as made apparent in Remark~\ref{rem:terms}, are closely related with terms.
Recall from Proposition~\ref{Fop_is_free_on_1} that $\bbF^\op$ is the free category with products on one object.
Spelled out, a Lawvere theory is a cartesian category $\clL$ and an identity-on-objects cartesian functor $\bbF^\op \to \clL$. A morphism of Lawvere theories is a functor $h :  \clL\to \clM$ s.t.\ the following triangle commutes:
\[
	\scriptsize\begin{tikzcd}
		& \bbF^\op \ar[dr, "q"] \ar[dl, "p"'] & \\
		\clL \ar[rr, "h"'] && \clM.
	\end{tikzcd}
\]
% Given the above, it is clear that $\bbF^\op$ is the initial object in $\textsf{Law}$. The terminal object of $\textsf{Law}$ coincides with the terminal prop, where there exactly one morphism between any two $m,n : \bbN$.

\begin{remark}\label{rem:terms}
	Every equational theory gives rise to a Lawvere theory.
	For the case of no equations $(\Sigma,\varnothing)$, this Lawvere theory $\clL_{\Sigma}$ is the
	free category with products on $\Sigma$. It also has a simple, concrete description that uses $\Sigma$-terms.
	An arrow $m\to n$ is an $n$-tuple
	\begin{equation}
		(t_1,\,t_2,\,\dots,\,t_n) \quad \text{where each}\quad t_i : T_{\Sigma}^{[m]} \label{eq:LSigma}
	\end{equation}
	i.e. where each term in the tuple may use formal variables from the set $\{1,\dots,m\}$. Composition
	of $(s_1,\dots, s_k) \from m \to k$ with $(t_1,\dots,t_n)\from k\to n$ is via substitution:
	\[(t_1[\substitute{s_1}{1},\dots \substitute{s_k}{k}],\dots, t_n[\substitute{s_1}{1},\dots,\substitute{s_k}{k}]) \from m \to n.\]
	Given a set of equations $E$, ordering the variables in each $s=t:E$ induces pairs of arrows $s,t: m\to 1$ (where $m$ is the number of variables appearing in $s$ and $t$). Then $\clL_{(\Sigma,E)}$ is obtained by ``equating'' $s$ and $t$ --this can be computed via a coequaliser, or directly as in~\eqref{eq:LSigma} where terms are taken modulo the smallest congruence containing the required equations. We omit the details.
\end{remark}
The Lawvere theory induced from the empty equational theory $(\varnothing,\,\varnothing)$ is $\bbF^\op$.

\subsection{Semantics for algebraic theories}\label{sem_for_alg}
Here we recall some of the basic elements of functorial semantics.
\begin{definition}[Model of a Lawvere theory]\label{def:lawveremodel}
	A \emph{model} for a Lawvere theory $\clL$ is a cartesian functor $L \from \clL \to \Set$.
	A model homomorphism $L \to L'$ is a natural transformation $\alpha \from L \Rightarrow L'$.
	This defines the category of models $\Mod_\clL$ of a Lawvere theory $\clL$.
\end{definition}

\begin{remark}\label{rem:forget}
	%Categories of models have
	There are forgetful functors $U: \Mod_\clL \to \Set$, given by evaluating on the terminal object $F \mapsto F(1)$. % Given a model,
	Intuitively, $U$ forgets the algebraic structure, returning the underlying carrier set.% We shall see that it is natural to equate classical varieties with pairs $(\Mod\clL,\,U)$.
\end{remark}

%This re-definition
Definition~\ref{def:lawveremodel}
is compatible with its classical counterpart: the data required to give a model of $(\Sigma,\,E)$, in the sense of Definition~\ref{def:model}, is precisely that required to give a functor $\clL_{(\Sigma,\,E)}\to\Set$.
%
%In this way, all the classes of Example~\ref{ex:varieties1} are instances of varieties in the above sense. %and \ref{ex_law_vars} below.
The functorial approach lends itself to generalisations: e.g.\ replacing $\Set$ with another cartesian category. Moreover, it allows for further structural analysis.
\begin{observation}
	$\Mod_\clL$ is closed under limits computed in the category of all functors $[\clL,\Set]$, because limits commute with limits. For a similar reason it is closed under sifted colimits. %because sifted colimits commute with finite products in $\Set$.
	Thus the inclusion $\Mod_\clL \hookrightarrow [\clL,\Set]$ \fix{creates} limits and sifted colimits\footnote{Sifted $\clJ$-indexed colimits satisfy the following property: given a functor $E : I \times \clJ \to \Set$, s.t.\ the category $I$ is discrete (namely, it is just a set), the following isomorphism holds:
		\[\textstyle\colim_\clJ \lim_I E(I,J)\cong \lim_I\colim_\clJ E(I,J).\]
		In other words, sifted colimits are those that commute with finite products in $\Set$.}.
\end{observation}

\begin{remark}[Multi-sorted, unsorted]\label{rem:sorts}\rm

	The codomain of $U\from\Mod_\clL \to \Set$ betrays that our presentation is \textit{single sorted}. Indeed when $\clL$ is $S$-sorted we obtain %in a similar way
	a functor $U\from \Mod_\clL \to [S,\Set]$. Historically, syntactic aspects are %were usually developed in a
	single sorted, while
	%setting,
	categorical variety theorems are most crisp in the \emph{unsorted} case. It is thus worthwhile to focus on these concepts in more detail.

	An $S$-sorted Lawvere theory is a cartesian $S$-coloured prop. Spelled out,
	it is an identity-on-objects cartesian functor $\bbF(S) \to \clL$, where $\bbF(S)$ is the free cartesian category on %the set of sorts
	$S$.

	An \emph{unsorted} Lawvere theory is simply a (small) category with products.
\end{remark}

In the remainder of this section we focus on the unsorted version, because the treatment is notationally and technically simplified. Nevertheless, much of the following is sort-agnostic.

\begin{observation}\label{obs:modeladjunction}
	A theory morphism $h\from \clL\ra \clL'$ induces a (contravariant)
	functor \[\Mod_{h}: \Mod_{\clL'} \to \Mod_{\clL}\] taking $F\from \clL' \to \Set$ to
	$F\circ h \from \clL \to \Set$. The functor $\Mod_h$ always admits a left adjoint:
	$\Mod_h$ preserves limits and sifted colimits because they are computed in the underlying functor category.
	The special adjoint functor theorem can now be used to obtain a left adjoint $L_h:\Mod_{\clL} \ra \Mod_{\clL'}$.
\end{observation}
\begin{example}
	For intuition, we consider a concrete example.
	Consider the inclusion $i$ of the theory of monoids in commutative monoids.
	Then $\Mod_i$ is the functor that ``forgets'' commutativity. % of the monoid operation.
	Its left adjoint takes a monoid and ``forces'' commutativity by %appropriately
	quotienting the underlying carrier set.
\end{example}

\begin{observation}\label{obs:freemodels}
	%The previous observation is enough to show that
	Lawvere theories have free models. Let $p: \bbF^\op \to \clL$ be a Lawvere theory.
	%Indeed we can look at it as a morphism in $\mathsf{Prod}$, then
	Observation~\ref{obs:modeladjunction} gives an adjunction %$F \dashv \Mod_{p}$
	\[F: \Mod_{\bbF^\op} \leftrightarrows \Mod_{\clL} : \Mod_{p}.\] Then $\Mod_{p}$ coincides with the forgetful functor of Remark~\ref{rem:forget}. The left adjoint $F$ gives %the construction of
	free objects.
\end{observation}

Because of Observation~\ref{obs:modeladjunction}, it is natural to take adjunctions as \emph{the} notion of variety morphism. Below, by \textit{unsorted}-variety we mean a category equivalent to $\Mod_\clL$ for $\clL$ with finite products.

\begin{definition}\label{morovar}
	Let $\clV, \clW$ be two (unsorted) varieties; a \emph{morphism of varieties}
	is a functor $R : \clV \to \clW$  satisfying the following:
	\begin{enumtag}{mv}
		\item $R$ admits a left adjoint $L : \clW \to \clV$;
		\item $R$ commutes with sifted colimits.
	\end{enumtag}
	Given that adjunctions compose, this data yields a category $\mathsf{Var}$.
\end{definition}

Let $\mathsf{Prod}$ be the 2-category whose objects are small cartesian categories, morphisms are cartesian functors and 2-cells are natural transformations. Then Observation~\ref{obs:modeladjunction} boils down to defining a 2-functor $\Mod\from \mathsf{Prod}^\op \to \mathsf{Var}$. The following captures the relationship between $\mathsf{Law}$ and $\mathsf{Var}$.

\begin{theorem}[\protect{\cite[Theorem 4.1]{adamek2003duality}}] \label{thm:lawvereadjunction}
	There exists a 2-adjunction whose unit is an equivalence:
	\[ \Th : \mathsf{Var}  \leftrightarrows  \mathsf{Prod}^{\text{op}} : \Mod \]
\end{theorem}

\begin{remark}
	One obtains the $S$-sorted version of Theorem~\ref{thm:lawvereadjunction} by slicing on both sides over the free category with products on $S$. This is given in more detail for \emph{partial} Lawvere theories in~\S\ref{gusection}.
\end{remark}

\subsection{Equational theories as monoidal theories}\label{subsec:termsasdiagrams}
Given that Lawvere theories are cartesian props, Theorem~\ref{thm:fox} suggests how to consider them as monoidal theories. We recall the recipe from~\cite{Bonchi2018}: the idea is to characterise $\Sigma$-terms as certain string diagrams, and then---through this lens---turn any equational theory into a monoidal theory.
\begin{recipe}\label{recipe:cartesian}
	Fix a signature $\Sigma$. A $\Sigma$-term $t: T_\Sigma^{[n]}$ is the same thing as
	a string diagram $n\to 1$ in the prop induced by the monoidal theory
	with
	\begin{itemize}
		\item
		      generators $\Gamma \Defeq \Sigma + \eqref{eq:comonoidgenerators}$
		\item
		      \eqref{eq:comonoid} together with equations that ensure naturality with respect to the comonoid structure. The latter
		      can be easily added as two additional equations for each $\sigma : \Sigma$:
		      \[\begin{tikzpicture}[xscale=.5, baseline=(current bounding box.center)]
				      \mor{\sigma}
				      \step{\comult}
				      \ezs{\comult\step{\umor{\sigma}\lmor{\sigma}}}{xshift=4cm}
				      \node at (2.5,2*\len) {$=$};
				      \ezs{\mor{\sigma}\step{\counit}
					      \ezs{\counit}{xshift=4cm}
				      }{xshift=9cm}
				      \node[left] at (0,2*\len) {$m$};
				      \node[left] at (4,2*\len) {$m$};
				      \node[left] at (9,2*\len) {$m$};
				      \node[left] at (13,2*\len) {$m$};
				      \node at (11.5,2*\len) {$=$};
			      \end{tikzpicture} \tag{SN$\sigma$}\label{eq:sigmanaturality}\]
	\end{itemize}
\end{recipe}
The Lawvere theory induced by equational theory $(\Sigma,E)$ can now be seen as the prop induced by
the monoidal theory $(\Gamma,F)$ where $F$ is the set of equations obtained by translating the equations in $E$ to string diagrams, together with \eqref{eq:comonoid}, and \eqref{eq:sigmanaturality} for each $\sigma : \Sigma$.

\smallskip
It is important to build an intuition behind this translation.
An obvious difference between terms and string diagrams is that the latter do not have named variables. The translation ensures that \emph{wires} play the role of variables,
and the comonoid structure plays the role of ``variable management''. We illustrate this with an example below.
\begin{example}\label{ex:propfromlawvere}
	The prop corresponding to the Lawvere theory induced by the equational theory of commutative monoids (Example~\ref{ex:eqtheorymonoids}) is the same as the prop of commutative bialgebra. For example, the term $m(m(x,x),y)$ in the theory of commutative monoids can be depicted as
	\[
		\begin{tikzpicture}[scale=.75,baseline=(current bounding box.center)]
			\up\comult
			\lid
			\step\lid
			\step[2]{\up{\cothingy{m}}}
			\step[3]{\cothingy{m}}
		\end{tikzpicture}
	\]

	In the term we have considered, the variable $x$ appears twice. In the corresponding diagram, the wire corresponding to $x$ starts with a comultiplication that witnesses the ``copying of $x$''.
\end{example}

\section{Algebra of partial maps}\label{sec:partialmaps}

We have seen that finite products are central in classical universal algebra. It is therefore natural to begin our development of its partial analogue by identifying the corresponding universal property in the partial setting. We will see that this amounts to replacing the class of cartesian categories with the class of \emph{discrete cartesian restriction categories} (DCR categories)~\cite{Coc12}. Next, we %identify the appropriate syntax by
characterise DCR categories in terms of algebraic structure, analogous to \ref{thm:fox} for cartesian categories.
%Finally we combine these pieces to obtain a notion of partial equational theory, partial Lawvere theory (Definition~\ref{partial_law}) and the associated (2-)category of varieties via a Lawvere-style functorial semantics (Definition~\ref{}).

\subsection{Partial functions}\label{sec:pardef}
The starting point of our journey is the (2-)category $\Par$ of sets and partial functions. Just as $\Set$ was the semantic universe for ordinary equational theories, $\Par$ is the semantic universe for partial equational theories.
We first recall an elementary, set theoretic presentation:
\begin{definition}\label{def:pardef}
	$\Par$ has sets as objects and \emph{partial functions} $f\from X\partialto Y$ as arrows, where
	a partial function $f$ is a pair $(\dom{f},\,\unf{f})$ where $\dom{f}\subseteq X$ is the \emph{domain of definition} of $f$ and $\unf{f}\from \dom{f}\to Y$ is a (total) function.
	Given a partial function $f\from X\partialto Y$, and some $X'\subseteq X$ we write $f\!\!\mid_{X'}$ for the partial function $(\dom{f}\cap X',\, f')$ where $f'\from \dom{f}\cap X'\to Y$ is $\unf{f}$ restricted to the (potentially smaller) domain of definition $\dom{f}\cap X'$. Similarly, given $Y'\subseteq Y$, write $f^{-1}(Y')=\{x\in \dom{f}\;\mid\; \unf{f}(x)\in Y'\}$.
	Given
	$f\from X\partialto Y$ and $g\from Y\partialto Z$, their composite is defined by
	$f \comp g = (f^{-1}(\dom{g}),\, (\unf f\!\mid_{f^{-1}(\dom{g})} \comp \unf g)$. The identity on $X$ is $(X,\eed_X)$.

	There is a natural partial order between partial functions $X\partialto Y$:
	\[
		f\leq g \quad \Defeq \quad
		\dom{f}\subseteq \dom{g} \;\wedge\; g \!\mid_{\dom{f}}\, = f. \]
	It is straightforward to verify that this data makes $\Par$ a category, and with $\leq$, a 2-category.
\end{definition}

Categorifying partiality has long history (see e.g.,~\cite{Rob88,Coc02}). We recall a classical approach:
\begin{definition}\label{defn:parcat}
	Suppose that $\mybb{C}$ has finite limits.
	Its 2-category of \emph{partial maps}, $\parcat{\mybb{C}}$ has:
	\begin{enumerate}
		\item \textbf{objects} are objects of $\mybb{C}$.
		\item \textbf{arrows} are equivalence classes $[m,f] : X \to Y$ of spans
		      $X \xleftarrow{m} A \xrightarrow{f} Y$
		      where $m$ is monic. We equate $(m,f) \sim (m',f')$ iff there is an isomorphism $\alpha$ s.t.\ the following diagram commutes:
		      \[
			      \scriptsize\begin{tikzcd}
				      & A \ar[dl,"m"'] \ar[drr,"f"'] \ar[r,"\alpha"] & A' \ar[dll,"m'"] \ar[dr,"f'"]  \\
				      X &&& Y.
			      \end{tikzcd}
		      \]
		\item \textbf{2-cells}: $[m,f]\leq [m',f']$ when there exists \emph{any} $\alpha$ that makes the diagram commute.
		\item \textbf{composition} is defined by pullback. Explicitly, the composite of $(m,f) : A \to B$ and $(m',g) : B \to C$ is the outer span of the diagram on the left
		      \[
						\begin{tikzcd}[row sep=small,column sep=small]
				      && X \wedge X' \ar[ld,"\pi_0"'] \ar[rd,"\pi_1"]\\
				      & X \ar[ld,"m"'] \ar[rd, "f"] && X' \ar[ld,"m'"'] \ar[rd,"g"] \\
				      A && B && C
			      \end{tikzcd}
			      % \hspace{1.5cm}
			      % \scriptsize\begin{tikzcd}
				    %   X \wedge X' \ar[d,"\pi_0"'] \ar[r,"\pi_1"] & X' \ar[d,"m'"]\\
				    %   X \ar[r,"f"'] & B
			      % \end{tikzcd}
		      \]
		      where the square with top $X\wedge X'$ is a pullback in $\mybb{C}$. Note that it doesn't matter which pullback, since any two choices will give isomorphic spans, and therefore equal morphisms.
		\item \textbf{Identities} are diagonal spans $(1_A,1_A) : A \to A$.
	\end{enumerate}
\end{definition}
Given a morphism $(m,f) : A \to B$ in $\mathsf{Par}(\mybb{C})$, we think of the monic $m : X \to A$ as a subobject, specifying which part of $A$ the morphism is defined on, and then $f : X \to B$ tells us what it does.

The following is a straightforward sanity check:
\begin{observation}
	There is an isomorphism of (2-)categories $\Par \cong \parcat{\Set}$.
\end{observation}

Just as a model of a total operation of arity $n$ is a function $A^n\ra A$ (an arrow in $\Set$), a model of a partial operation ought to be a partial function $A^n\partialto A$ (an arrow in $\Par$). For this reason, it is important to understand the mathematical status of the cartesian product in $\Par$. Interestingly, $\Par$ has categorical products, but these \emph{do not} correspond to the cartesian product of sets.\footnote{The categorical product of $A$ and $B$ in $\Par$ is $(A+\{\star\})\times(B+\{\star\}) - \{(\star,\,\star)\}$. This can be seen via the equivalence $1/\Set\simeq \Par$. Limits in the coslice category $1/\Set$ are calculated pointwise,
	and the functor $1/\Set \to \Par$ removes the point}
\subsection{Cartesian Restriction Categories}

It is by focusing on the universal property of the cartesian product in $\Par$ that we are able to identify a  generalisation of Lawvere's approach to partial operations. This is the goal of this section.

\emph{Restriction categories} were devised to study partial phenomena in an axiomatic setting. Here we give a whirlwind tour, more details can be found in~\cite{Coc02,Coc07,Coc12}. % and Appendix~\ref{rc_basic_defs}.
In a restriction category, every arrow $f : A \to B$ has an associated idempotent $\rest{f} : A \to A$, thought of as the identity function restricted to the domain of definition of $f$. We call them %$\rest{f}$
\emph{domain idempotents}.
%, and they must satisfy a number of axioms.
Arrows %$f : A \to B$ in a restriction category with
where $\rest{f} = 1_A$ are called \emph{total}, and form a subcategory. Further, we have:

\begin{remark}\label{rem:enriched}
	Any restriction category is poset-enriched, with the ordering defined by
	\[
		f \leq g \Leftrightarrow \rest{f}\comp g = f
	\]
\end{remark}

%It is often desirable for functors between restriction categories to
Functors $F$ that preserve domain idempotents ($F\rest{f} = \rest{Ff}$) are called \emph{restriction functors}. Restriction categories and restriction functors form a category. This extends to a 2-category in which the 2-cells are \emph{lax transformations}. A lax transformation $\alpha : F \to G$ of restriction functors $F,G : \bbX \to \bbY$ consists of a family of total maps $\alpha_A : FA \to GA$ in $\bbY$ indexed by the objects $A$ of $\bbX$ s.t.\ for every $f : A \to B$ of $\bbX$ the usual naturality square \emph{commutes up to inequality}:
\[
	\scriptsize\begin{tikzcd}
		FA \ar[d,"Ff"'] \ar[r,"\alpha_A"] & GA \ar[d,"Gf"] \\
		FB \ar[r,"\alpha_B"'] \ar[ur,phantom,"\leq"] & GB
	\end{tikzcd}
\]
where $\leq$ is the ordering introduced above. Call this 2-category $\mathsf{RCat}^\leq$. Just as categories with finite products enjoy a universal property in the 2-category $\Cat$, those with finite restriction products have a universal property in $\mathsf{RCat}^\leq$. In general, formal limits in $\mathsf{RCat}^\leq$ are called \emph{restriction limits}~\cite{Coc07}. A \emph{cartesian restriction} (CR) category is a restriction category with finite restriction products.

\begin{observation}[\cite{Coc02}]
	$\Par$ is a CR category, with the cartesian product as restriction product.
\end{observation}

CR categories have appeared in the literature under a variety of different names, including \emph{p-category with a one-element object}~\cite{Rob88} and \emph{partially cartesian category}~\cite{Cur89}. For our development, it is crucial that the data of CR categories can be equivalently captured as follows:

\begin{theorem}[\cite{Coc07}]\label{thm:newfox}
	A CR category is the same thing as a symmetric monoidal category where every object is equipped with a commutative comonoid structure that is \eqref{eq:coherent} and
	the comultiplication is natural. That is, for any $f : A \to B$ we have $f \comp \delta_B = \delta_A \comp (f \otimes f)$.
  %   A cartesian category is the same thing as a symmetric monoidal category
	%	where every object can be equipped with a \eqref{eq:coherent} and \eqref{eq:natural} commutative comonoid structure.
	%
	%A CR category is exactly a symmetric monoidal category equipped with a coherent commutative comonoid %structure $(A,\delta_A,\varepsilon_A)$ for each object $A$, s.t.\
  %
	%A symmetric monoidal category is a CR category if and only if
	%every object can be equipped with a commutative comonoid structure that is \eqref{eq:coherent} and
	%the comultiplication is natural. That is, for any $f : A \to B$ we have $f \comp \delta_B = \delta_A \comp (f \otimes f)$.
\end{theorem}
%\begin{proof}
%	In Appendix~\ref{app:a_crc_is}.
%\end{proof}

From this perspective a CR category is very close to a cartesian category viewed as a monoidal category through \ref{thm:fox}. The difference is that we do not ask for the counit of the comonoid to be natural. This has profound consequences: for instance, the same symmetric monoidal category may have more than one such chosen comonoid structure, thus definining different CR categories.

Given the algebraic data, the domain idempotent $\rest{f} : A \to A$ of an arrow $f : A \to B$ in a CR category is recovered as:
\[
	\begin{tikzpicture}[xscale=.5]
		\comult
		\ezs{\umor{f}}{xshift=1cm}
		\ezs{\counit}{xshift=2cm, yshift=\len cm}
		\ezs{\lid}{xshift=1cm,xscale=2}
		\node[left] at (-1,2*\len) {$\overline f$};
		\node[left] at (0,2*\len) {$=$};
	\end{tikzpicture}
\]
and so in particular the subcategory of $\bbX$ for which the counit is natural is precisely the subcategory of total maps. Notice that this means the subcategory of total maps of a CR category is cartesian.

\begin{definition}
	A \emph{CR functor} between two CR categories $F : \bbX \to \bbY$ is a functor that preserves the %cartesian restriction
	algebraic	structure. That is, $F(A\otimes B) = FA \otimes FB$, $F1 = 1$, $F\delta_A = \delta_{FA}$ and $F\varepsilon_A = \varepsilon_{FA}$.
\end{definition}

\begin{remark}\label{rem:laxtotal}
	A lax transformation of CR functors may be equivalently specified as a family of maps $\alpha_A : FA \to GA$ indexed by the objects $A$ of $\bbX$ s.t.\ for every $f : A \to B$ we have $Ff \comp \alpha_B \leq \alpha_A \comp Gf$. We do not need to ask that each $\alpha_A$ is total, since if $F$ and $G$ preserve the cartesian restriction structure, then they are \emph{automatically} total. In particular the diagram on the left gives the inequality on the right, which gives that $\alpha_A$ is total:
	\[
		\scriptsize\begin{tikzcd}
			FA \ar[d,"F\varepsilon_A"'] \ar[r,"\alpha_A"] & GA \ar[d,"G\varepsilon_A"] \\
			FI \ar[r,"\alpha_I"'] \ar[ur,phantom,"\leq"]& GI
		\end{tikzcd}
		\qquad\qquad
		\graffle[0pt]{.3cm}{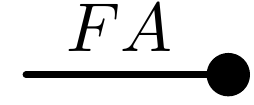}
		\;\leq\;
		\graffle[5pt]{.5cm}{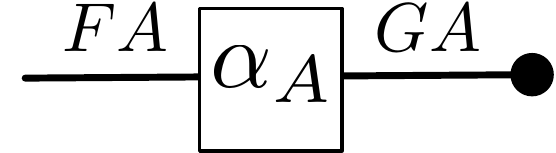}
	\]
\end{remark}

\subsection{Discrete Cartesian Restriction Categories}
Restriction products do not quite capture all the properties of $\Par$ needed for partial universal algebra. In particular, we require CR categories with the following extra structure:
\begin{definition}
	A CR category is said to be $\emph{discrete}$ (DCR category~\cite{Coc12}) if for each object $A$ there is an arrow $\mu_A : A \otimes A \to A$ that is \emph{partial inverse} to $\delta_A$. That is, $\delta_A \comp \mu_A = \rest{\delta_A} = 1_A$ and $\mu_A \comp \delta_A = \rest{\mu_A}$.
\end{definition}

We give a novel presentation of DCR categories, inspired by the work of \cite{Gil14}. Central to our presentation is the notion of a commutative special Frobenius algebra in which the monoid does not have a unit, which we call a \emph{partial Frobenius algebra}. More precisely:

\begin{definition}
	A \emph{partial Frobenius algebra} $(A,\delta_A,\mu_A,\varepsilon_A)$ in a symmetric monoidal category consists of a commutative comonoid $(A,\delta_A,\varepsilon_A)$ and a commutative semigroup $(A,\mu_A)$ s.t.\ $(A,\delta_A,\mu_A)$ is a semi-Frobenius algebra. Diagramatically, this is the comonoid structure we have already seen together $\mu_A$, which we depict as $\raisebox{.25em}{\begin{tikzpicture}[scale=.75,baseline=(current bounding box.center)]
\mult
\end{tikzpicture}}$ in our string diagrams, subject to the following equations:
	\begin{equation} \label{eq:multassociative} \tag{MCA}
		\raisebox{.5em}{\begin{tikzpicture}[xscale=.75,baseline=(current bounding box.center)]
				\umult
				\lid
				\step{\mult}
				\ezs{\lmult
					\uid
					\step{\mult}}{xshift=3cm,yshift=\len cm}
				\node at (2.5,2*\len) {$=$};
			\end{tikzpicture}}\hspace{2em}
		\raisebox{.5em}{\begin{tikzpicture}[xscale=.75,baseline=(current bounding box.center)]
				\braid\step{\mult}\ezs{\mult}{xshift=3cm}
				\node at (2.5,2*\len) {$=$};
			\end{tikzpicture}}
	\end{equation}
	\vspace{-.8cm}
	\begin{equation}
		\label{eq:sfrob} \tag{SFROB}
		\begin{tikzpicture}[xscale=.75,baseline=(current bounding box.center)]
			\ezs{\id}{yshift=2*\len cm}\lcomult
			\step{\umult\ezs{\lid}{yshift=-\len cm}}
			\ezs{\mult\step{\comult}}{xshift=3cm}
			\ezs{\ezs{\id}{yshift=2*\len cm}\lcomult
				\step{\umult\ezs{\lid}{yshift=-\len cm}}}{xshift=8cm, xscale=-1}
			\node at (2.5,2*\len) {$=$};
			\node at (5.5,2*\len) {$=$};
			\ezs{\comult\step{\mult}
				\ezs{\id}{xshift=3cm}
				\node at (2.5,2*\len) {$=$};}{xshift=10cm}
		\end{tikzpicture}
	\end{equation}
\end{definition}
Note that there is some redundancy in the equational presentation above, as discussed in~\cite{carboni1991matrices}.
We now extend the characterisation of CR categories given in Theorem~\ref{thm:newfox} to DCR categories:
\begin{theorem}\label{thm:dcrc}
  A DCR category is the same thing as a symmetric monoidal category where every object $A$ is equipped
	with a coherent partial Frobenius algebra structure $(A,\delta_A,\varepsilon_A,\mu_A)$ s.t.\ the comultiplication is natural. That is, for any $f : A \to B$ we have $f \comp \delta_B = \delta_A \comp (f \otimes f)$.
  %A CR category is the same thing as a symmetric monoidal category where every object is equipped with a commutative comonoid structure that is \eqref{eq:coherent} and the comultiplication is natural. That is, for any $f : A \to B$ we have $f \comp \delta_B = \delta_A \comp (f \otimes f)$.
	%
	%A symmetric monoidal category is a DCR category if and only if
	%every object can be equipped with a coherent partial Frobenius algebra structure $(A,\delta_A,\varepsilon_A,\mu_A)$ for each object $A$ s.t.\ the comultiplication is natural. That is, for any $f : A \to B$ we have $f \comp \delta_B = \delta_A \comp (f \otimes f)$.
	%
	%	commutative comonoid structure that is \eqref{eq:coherent} and
	%
	%A DCR category is precisely a symmetric monoidal category equipped with a coherent partial Frobenius algebra structure $(A,\delta_A,\varepsilon_A,\mu_A)$ for each object $A$ s.t.\ the comultiplication is natural. That is, for any $f : A \to B$ we have $f \comp \delta_B = \delta_A \comp (f \otimes f)$.
\end{theorem}
%\begin{proof}
%	In \ref{app:a_crc_is}.
%\end{proof}

DCR categories are intimately connected to categories with finite limits \cite{Coc12}. In particular:
\begin{proposition}
	If $\bbC$ is a category with finite limits, $\mathsf{Par}(\bbC)$ is a DCR category.
\end{proposition}
%\begin{proof}
%	In \ref{dut}.
%\end{proof}
\begin{definition}[the 2-category $\mathsf{DCRC}^\leq$]\label{the_def_of_dcrc}
	It follows that for any CR functor $F : \bbX \to \bbY$ between DCR categories, we have $F\mu_A = \mu_{FA}$. CR functors therefore give the notion of morphism between DCR categories. We consider the strict 2-category of DCR categories, restriction functors, and lax transformations, which we call $\mathsf{DCRC}^\leq$.
\end{definition}

\section{Partial Lawvere theories}\label{parlawve}

In this section we develop a Lawvere-style approach to \emph{partial algebraic theories}, where %\textit{(i)}
operations may be partial. % defined only on a proper subset of their domain.
%, and \textit{(ii)} equations relate terms constructed from such operations.
Ordinary Lawvere theories are determined by % a bijective on object functor from
the free cartesian category on a single object $\bbF^\op$; we are thus interested in the analogue of $\bbF^\op$ in the world of DCR categories.
\subsection{The free DCR category on one object}\label{the_free_dcr_on}

Given Theorem~\ref{thm:dcrc}, we have an explicit description for the DCR category on one object: it
is the prop $\bbP\bbF$ induced from the monoidal theory of partial Frobenius algebras. That is, it
has generators $\{\,\raisebox{.25em}{\begin{tikzpicture}[scale=.75,baseline=(current bounding box.center)]
\mult
\end{tikzpicture}} \; , \; \raisebox{.25em}{\begin{tikzpicture}[scale=.75,baseline=(current bounding box.center)]
\comult
\end{tikzpicture}} \; , \; \raisebox{.25em}{\begin{tikzpicture}[xscale=.5,baseline=(current bounding box.center)]
\counit
\end{tikzpicture}} \,\}$
and equations~\eqref{eq:multassociative}, \eqref{eq:comonoid} and~\eqref{eq:sfrob}.

It turns out that one can gain a precise intuition on what $\bbP\bbF$ looks like by mimicking the way in which the props $\bbF$ and its opposite $\bbF^\op$ describe familiar algebraic structures. In fact, the prop $\bbC\bbM$ of commutative monoids is isomorphic to the prop $\bbF$ (see Observation~\ref{obs:Fmonoids}), and similarly, the prop $\bbC\bbC$ of commutative comonoids is isomorphic to $\bbF^\op$ (Observation~\ref{obs:Fopcomonoids}).

The prop $\bbP\bbF$ of partial Frobenius algebras that we want to describe here can be given a similar ``combinatorial'' characterisation relying on the insights of Lack~\cite{compprops}.
First, we note that the prop $\bbC\bbA\bbM$ induced by the monoidal theory of commutative semigroups ($\{\,\raisebox{.25em}{\begin{tikzpicture}[scale=.75,baseline=(current bounding box.center)]
\mult
\end{tikzpicture}}\,\}$ and equations~\eqref{eq:multassociative}) is isomorphic to sub-prop $\bbF_s\subset\bbF$ of finite sets and \emph{surjective functions}.
\begin{observation}\label{obs:surjectivefunctions}
	As props, $\bbC\bbA\bbM \cong \bbF_s$.
\end{observation}
This is intuitive: as observed in Remark~\ref{rem:pictures}, string diagrams of $\bbC\bbM$ allow one to ``draw'' all functions $[m]\to [n]$. Doing without the unit means that we can express exactly the surjective ones.
%\fos{Here reference to ``picture of functions''}

Next, we know from~\cite{compprops} that the prop $\bbF\bbR\bbO\bbB$ induced by the monoidal theory of special Frobenius algebras with generators $\{\, \raisebox{.25em}{\begin{tikzpicture}[scale=.75,baseline=(current bounding box.center)]
\mult
\end{tikzpicture}} \; , \; \raisebox{.25em}{\begin{tikzpicture}[xscale=.5,baseline=(current bounding box.center)]
\unit
\end{tikzpicture}} \; , \; \raisebox{.25em}{\begin{tikzpicture}[scale=.75,baseline=(current bounding box.center)]
\comult
\end{tikzpicture}} \; , \; \raisebox{.25em}{\begin{tikzpicture}[xscale=.5,baseline=(current bounding box.center)]
\counit
\end{tikzpicture}} \,\}$
and equations~\eqref{eq:monoid}, \eqref{eq:comonoid} and \eqref{eq:sfrob} is isomorphic to the prop of \emph{cospans} of finite sets $\mathsf{Cospan}(\bbF)$. An arrow $m\ra n$ here is (an isomorphism class of) a cospan of functions
$[m] \xrightarrow{f} [k] \xleftarrow{g} [n]$, and composition is by pushout.
\begin{proposition}[\protect{\cite{compprops}}]\label{pro:Frobeniuscospans}
	As props, $\bbF\bbR\bbO\bbB \cong \mathsf{Cospan}(\bbF)$.
\end{proposition}

Given that surjective functions are closed under composition and pushouts in $\bbF$,
we can consider the subprop $\mathsf{Cospan}_s(\bbF)$ of  $\mathsf{Cospan}(\bbF)$ with arrows those cospans where the left leg is surjective. Now, combining Observation~\ref{obs:surjectivefunctions} and Proposition~\ref{pro:Frobeniuscospans} yields:
\begin{proposition}
	As props, $\bbP\bbF \cong \mathsf{Cospan}_s(\bbF)$.
\end{proposition}

This gives us a combinatorial characterisation of $\bbP\bbF$. But there is a more familiar and satisfactory way of describing $\mathsf{Cospan}_s(\bbF)$. Given that cospans in $\bfC$ are spans in $\bfC^\op$, and epimorphisms in $\bfC$ are monomorphisms in $\bfC^\op$, we see that $\mathsf{Cospan}_s(\bbF) = \parcat{\bbF^\op}$,
since a cospan in $\bbF$ with left leg surjective is the same thing as a span in $\bbF^{\op}$ with left leg a monomorphism.
Therefore, we obtain:
\begin{proposition}\label{as_props_ultimo}
	As props, $\bbP\bbF\cong \parcat{\bbF^\op}$.
\end{proposition}
\subsection{Partial Lawvere theories}\label{subsec:partialLawvere}

We have seen that $\bbF^\op$ is central to the definition of Lawvere theories, being the free cartesian category on one object. The prop $\parcat{\bbF^\op}$, being the free DCR on one object, plays the corresponding role in the definition of partial Lawvere theories.

\begin{definition}\label{partial_law}
	A partial Lawvere theory $\clL$ is a DCR prop.
\end{definition}
Spelled out, a partial Lawvere theory is a DCR category $\clL$ for which there is an identity-on-objects CR functor $\parcat{\bbF^\op} \to \clL$. A morphism of partial Lawvere theories is a functor $h :  \clL\to \clM$ s.t.\ the following triangle commutes:
\[
	\scriptsize\begin{tikzcd}
		& \parcat{\bbF^\op} \ar[dr, "q"] \ar[dl, "p"'] & \\
		\clL \ar[rr, "h"'] && \clM.
	\end{tikzcd}
\]

This defines the category $\pLaw$ of partial Lawvere theories.

Mimicking also the definition of model of a Lawvere theory, we obtain at once the notion of model of a \emph{partial} Lawvere theory:
\begin{definition}[Model of a partial Lawvere theory]\label{models_of_parlaw}
	A \emph{model} for a partial Lawvere theory $\clL$ is a CR functor $L \from \clL \to \Par$.
	A homomorphism $L \to L'$ is a \emph{lax} natural transformation $\alpha \from L \Rightarrow L'$.
\end{definition}
\begin{definition}\label{cat_of_models_of_parlaw}
	The category of models and homomorphisms of a partial Lawvere theory $\clL$ is denoted
	$\ct{pMod}_\clL$. As explained in Remark~\ref{rem:laxtotal}, the homomorphisms are \emph{total} functions.
\end{definition}

\section{Partial equational theories}\label{sec:partialequationaltheories}

In order to consider interesting examples of partial Lawvere theories, we need to introduce the notion of \emph{partial equational theory}. For partial structures, these are the syntactic analogue of equational theories, and yield partial Lawvere theories in a similar way to how equational theories yield Lawvere theories.

%We have seen (Definition~\ref{monosigna}) that
Monoidal signatures (Definition~\ref{monosigna}) $\Gamma$ have unrestricted arities and coarities.
Instead, a signature $\Sigma$ of an equational theory (Definition~\ref{unialg}) has function symbols of arbitrary arities, but---considered as a monoidal signature---all coarities are 1. Partial signatures are an intermediate concept: as for equational theories, coarities $>1$ are redundant, but we need to admit symbols of coarity $0$.

\begin{definition}
	A partial signature $\Delta\Defeq \Delta_0+\Delta_1$, where $\Delta_0$ is the set of operations of coarity $0$, and $\Delta_1$ is the set of operations of coarity 1. Each $\delta : \Delta_i$ comes with an arity $\text{ar}(\delta) : \bbN$.
\end{definition}

Differently from ordinary equational theories, we cannot use classical terms, which---as discussed in Remark~\ref{rem:terms}---are tied to an underlying cartesian structure.
%\fos{(but also Ivan:) what is the obstruction to use classical terms?}
Instead, we adapt Recipe~\ref{recipe:cartesian} to DCR categories, obtaining \emph{partial} terms as particular string diagrams. Before we launch into formal definitions, and illustrate them with a variety of examples, let us establish some intuitions for how to read the string diagrams.
\begin{itemize}
\item string diagrams represent partial terms, obtained through composing partial operations,
\item equalities and inequalities between them are understood in the sense of Kleene,
\item the comonoid structure $\{\raisebox{.25em}{\begin{tikzpicture}[scale=.75,baseline=(current bounding box.center)]
\comult
\end{tikzpicture}},\,\raisebox{.25em}{\begin{tikzpicture}[xscale=.5,baseline=(current bounding box.center)]
\counit
\end{tikzpicture}}\}$ plays a similar role to that described
\S\ref{sec:history}
\end{itemize}

\begin{recipe}\label{recipe:dcr}\rm
	Given a partial signature $\Delta$, the free DCR prop $\clL_\Delta$ on $\Delta$ is the prop induced from the monoidal theory with
	\begin{itemize}
		\item generators $(\Delta + \{\raisebox{.25em}{\begin{tikzpicture}[scale=.75,baseline=(current bounding box.center)]
\mult
\end{tikzpicture}},\,\raisebox{.25em}{\begin{tikzpicture}[scale=.75,baseline=(current bounding box.center)]
\comult
\end{tikzpicture}},\,\raisebox{.25em}{\begin{tikzpicture}[xscale=.5,baseline=(current bounding box.center)]
\counit
\end{tikzpicture}}\})$
		\item equations~\eqref{eq:multassociative}, \eqref{eq:comonoid} and~\eqref{eq:sfrob}, together with
		      \[ 
						\begin{tikzpicture}[xscale=.75,baseline=(current bounding box.center)]
				      \mor{\delta}
				      \step{\comult}
				      \ezs{\comult\step{\umor{\delta}\lmor{\delta}}\node[left] at (0,2*\len) {$m$};
				      }{xshift=3.5cm}
				      \node[left] at (0,2*\len) {$m$};
				      \node[above] at (2.5,\len) {$=$};
						\end{tikzpicture}
						\]
		      for each $\delta : \Delta_1$, where $\text{ar}(\delta)=m$.
	\end{itemize}
\end{recipe}

\begin{definition}[Partial Equational Theory]
	A \emph{partial equation} is a pair $(s,t)$ where $s,t : \clL_\Delta(m,n)$
	for some $m,n:\bbN$; we usually write `$s=t$'.
	A \emph{partial equational theory} is a pair $(\Delta,G)$ where $\Delta$ is a partial signature and $G$ is a set of partial equations.
\end{definition}

%We shall consider a selection of examples of partial equational theories in \S\ref{sec:examples}. First,
We first return to a familiar example.
\begin{example}[(Partial) Commutative Monoids]
	We start with the monoidal theory of commutative monoids (Example~\ref{ex:propmonoid}), where the multiplication and unit generators are re-coloured to red to avoid a clash. %The difference between ordinary commutative monoids and partial commutative monoids defined by this partial theory is that
	In models, the multiplication operation may be partially defined \emph{and} the unit may be undefined. To define the partial theory of \emph{total} commutative monoids, we'd need to add equations:
	\begin{equation}\label{eq:totality}
		\begin{tikzpicture}[xscale=.75,baseline=(current bounding box.center)]
			\emult\step{\counit}
			\ezs{
				\up\counit
				\down\counit
			}{xshift=4cm}
			\eql{2.75}
		\end{tikzpicture}
		\hspace{3em}
		\begin{tikzpicture}[xscale=.75,baseline=(current bounding box.center)]
			\eunit\step{\counit} \ezs{\node at (0,2*\len) {$\raisebox{.25em}{\begin{tikzpicture}[xscale=.5,baseline=(current bounding box.center)]
\counit
\end{tikzpicture}}$};}{xshift=3cm}
			\eql{2.25}
		\end{tikzpicture}
	\end{equation}
\end{example}

\begin{example}[Equational Theories]
	Any equational theory $(\Sigma,E)$ is an example.
	One follows Recipe~\ref{recipe:dcr}, adding equations analogous to~\eqref{eq:totality} to specify that every generator in $\Sigma$ is total. The category of models of this partial theory then agrees with that of the Lawvere theory $\clL_{(\Sigma,E)}$.
\end{example}

The following elementary examples illustrate the novel features of partial Lawvere theories, highlighting the way in which they differ from classical (i.e., total) Lawvere theories.
\begin{example}[Equivalence Relations]
	Consider the partial Lawvere theory consisting of a single binary operation $R$ with coarity $0$, together with equations expressing symmetry and reflexivity:
	\[
		\begin{tikzpicture}[xscale=.75,baseline=(current bounding box.center)]
			\theR
		\end{tikzpicture}:\hspace{1.5em}
		\begin{tikzpicture}[xscale=.75,baseline=(current bounding box.center)]
			\braid\step{\theR}
			\ezs{\theR}{xshift=3cm}
			\eql{2.5}
		\end{tikzpicture}\hspace{3em}
		\begin{tikzpicture}[xscale=.75,baseline=(current bounding box.center)]
			\comult\step\theR
			\ezs{\counit}{xshift=3cm}
			\eql{2.5}
		\end{tikzpicture}  \]
	Note that \emph{inequations} of terms, as in Remark \ref{rem:enriched}, do not add expressivity. As such, we may use them freely when specifying partial Lawvere theories. Transitivity is intuitively captured by the inequation on the left, which, unfolding the definition of $\leq$, is precisely the equation on the right:
	\[
		\begin{tikzpicture}[xscale=.75,baseline=(current bounding box.center)]
			\up[2]\uid
			\down[2]\lid
			\comult
			\step{\up[2]\theR
				\down[2]\theR}
			\up{\ezs{
					\uid
					\step\theR
					\down[2]\braid
					\step{\down[2]\lcounit}
				}{xshift=3cm}}
			\node at (2.5,2*\len) {$\le$};
		\end{tikzpicture}
		\hspace{2.5cm}
		\begin{tikzpicture}[xscale=.75,baseline=(current bounding box.center)]
			\up[2]\uid
			\down[2]\lid
			\comult
			\step{\up[2]\theR
				\down[2]\theR}
			\ezs{\up[3]\comult
				\down[3]\comult
				\comult
				\step{\up[3]\uid
					\up[2]\turn
					\down[3]\lid
					\ezs{\down[1.4]\coturn}{yscale=2.5}
					\turn
					\down{\turn}
				}
				\step[2]{\ezs{\theR}{yshift=1.125cm, yscale=.5}}
				\step[2]{\ezs{\theR}{yshift=.375cm, yscale=.5}}
				\step[2]{\ezs{\theR}{yshift=-.625cm, yscale=.5}}
			}{xshift=3cm}
			\eql{2.5}
		\end{tikzpicture}
	\]
	A model $\mathcal{A}$ of this theory consists of a set $A$ together with an equivalence relation $=_A \subseteq A \times A$ corresponding to the domain of definition of $\mathcal{A}(R)$. A morphism $F : \mathcal{A} \to \mathcal{B}$ is a function $F : A \to B$ with $a =_A b \Rightarrow Fa =_B Fb$, which arises from the requirement that $F$ is a lax transformation:
	\begin{equation*}
		\begin{tikzpicture}[baseline=(current bounding box.center)]
			\theR
			\up{\gau{A}}
			\down{\gau{A}}
			\ezs{
				\ezs{\umor{F}\lmor{F}}{xscale=.75}
				\up{\gau{A}}
				\down{\gau{A}}
				\step[.75]\theR
			}{xshift=2cm}
			\node at (1.25,2*\len) {$\le$};
		\end{tikzpicture}
	\end{equation*}
	Thus, the variety corresponding to this theory is the category of \emph{Bishop sets} (\emph{setoids})~\cite{Palmgren09constructivistand}.
\end{example}
\begin{example}[Partial Combinatory Algebras]
	A \emph{partial combinatory algebra} (PCA) is a set $A$ with a binary partial operation $\_\bullet\_ : A \times A \to A$, and elements $\mathsf{s},\mathsf{k} \in A$ s.t.\ for any $x,y,z \in A$:
	\begin{enumerate}[(i)]
		\item $(\mathsf{k} \bullet x) \bullet y = x$
		\item $((\mathsf{s} \bullet x) \bullet y) \bullet z = (x \bullet z) \bullet (y \bullet z)$
		\item $(\mathsf{s} \bullet x) \bullet y$ is defined
	\end{enumerate}
	where ``='' is Kleene equality. The partial Lawvere theory of PCAs has three generators:
	\[
		\begin{tikzpicture}
			\dcouni{k}
			\ezs{\dcouni{s}}{xshift=3cm}
			\ezs{\emult}{xshift=6cm}
		\end{tikzpicture}
	\]
	and equations that ensure the totality of $k,s$, i.e.\ they define elements of the carrier, and $(iii)$:
	\[
		\begin{tikzpicture}[baseline=(current bounding box.center)]
			\ezs{
				\dcouni{k}
				\ezs{\counit}{xshift=1cm}}{xscale=.75}
			\ezs{\akasa}{xshift=2.5cm}
			\node at (1.75,2*\len) {$=$};
		\end{tikzpicture}
		\qquad
		\begin{tikzpicture}[baseline=(current bounding box.center)]
			\ezs{
				\dcouni{s}
				\ezs{\counit}{xshift=1cm}}{xscale=.75}
			\ezs{\akasa}{xshift=2.5cm}
			\node at (1.75,2*\len) {$=$};
		\end{tikzpicture}
		\qquad
		\begin{tikzpicture}[baseline=(current bounding box.center),xscale=.75,yscale=-1]
			\ldcouni{s}\uid
			\ezs{\emult}{xshift=1cm}
			\ezs{\lid}{yshift=.75cm,xscale=2}
			\ezs{\emult}{xshift=2cm,yshift=.25cm}
			\ezs{\counit}{xshift=3cm,yshift=.25cm}
			\ezs{\counit\ezs{\counit}{yshift=.5cm}}{xshift=5cm}
			\node at (4.25,3*\len) {$=$};
		\end{tikzpicture}
	\]
	as well as equations for $(i)$ and $(ii)$:
	\[
		\begin{tikzpicture}[xscale=.75,yscale=-1,baseline=(current bounding box.center)]
			\ldcouni{k}\uid
			\ezs{\emult}{xshift=1cm}
			\ezs{\lid}{yshift=.75cm,xscale=2}
			\ezs{\emult}{xshift=2cm,yshift=.25cm}
			\ezs{\lid}{xshift=4.5cm}
			\ezs{\counit}{yshift=.25cm,xshift=4.5cm}
			\node at (3.75,2*\len) {$=$};
		\end{tikzpicture}
		\qquad\qquad\qquad
		\begin{tikzpicture}[xscale=.5,yscale=-1,baseline=(current bounding box.center)]
			\ldcouni{s}\uid
			\ezs{\emult}{xshift=1cm}
			\ezs{\emult}{xshift=2cm, yshift=.25cm}
			\ezs{\emult}{xshift=3cm, yshift=.5cm}
			\ezs{\id}{yshift=.5cm,xscale=2}
			\ezs{\id}{yshift=.75cm,xscale=3}
			\ezs{\twoid
				\ezs{\comult}{yshift=1cm}
				\ezs{\braid}{xshift=1cm,yshift=.5cm}
				\ezs{\lid}{xshift=1cm}
				\ezs{\uid}{xshift=1cm,yshift=1cm}
				\ezs{\emult}{xshift=2cm}
				\ezs{\emult}{xshift=2cm,yshift=1cm}
				\ezs{\emult}{xshift=3cm,yscale=2}
			}
			{xshift=6cm, xscale=1.25}
			\node at (5,1) {$=$};
		\end{tikzpicture}
	\]
	The variety here is the category of PCAs and homomorphisms preserving the applicative structure.
\end{example}

\begin{example}[Pairing Functions]
	Consider the partial Lawvere theory with two operations which we think of as \emph{pairing} and \emph{unpairing} respectively, subject to the equation on the right:
	% $
	% 	\raisebox{.25em}{\begin{tikzpicture}[scale=.5,baseline=(current bounding box.center)]
	% 		\ecomult
	% 		\ezs{\emult}{xshift=2cm}
	% 	\end{tikzpicture}}

	\[
		\begin{tikzpicture}[baseline=(current bounding box.center)]
			\emult
			\step[1.5]{\ecomult}
			\ezs{
				\emult
				\ezs{\ecomult}{xshift=1cm}
				\ezs{\twoid}{xshift=3cm}
				\node at (2.5,2*\len) {$=$};}{xshift=6cm}
		\end{tikzpicture}
	\]
	Models are sets equipped with a \emph{pairing function}, and model morphisms map pairs to pairs. For example, $\mathbb{N}$ and Cantor's pairing function, or $\Lambda$ -- the set of untyped $\lambda$-terms -- with the usual pairing and projection functions. Note that our equation makes pairing a section, and so it is total.
\end{example}

\section{Multi-Sorted Equational Theories}\label{sec:multi}
In this section we present a progression of multi-sorted partial Lawvere theories for categories with different kinds of structure. While our development of partial Lawvere theories has thus far focused on the single-sorted case, the move to multi-sorted theories contains no surprises, so we omit the details. The short version is that props are replaced with \emph{coloured} props, and the sorting discipline changes accordingly. The examples that follow are developed incrementally: Each step adds more categorical structure to the models by adding the appropriate operations and equations to the theory, culminating in the partial Lawvere theory of cartesian closed categories.

\begin{example}[Directed Graphs]\label{ex:dirgraph}
	We begin with the partial Lawvere theory of \emph{directed graphs}, which has a sort $O$ of vertices and a sort $A$ of edges, together with source and target operations:
	\[
		\begin{tikzpicture}[baseline=(current bounding box.center)]
			\mor{s}
			\node[left] at (0,2*\len) {$\scriptscriptstyle A$};
			\node[right] at (1,2*\len) {$\scriptscriptstyle O$};
			\ezs{\mor{t}
				\node[left] at (0,2*\len) {$\scriptscriptstyle A$};
				\node[right] at (1,2*\len) {$\scriptscriptstyle O$};
			}{xshift=3cm}
			\ezs{\ezs{\mor{s}\step{\counit}}{xscale=.5}\ezs{\counit}{xshift=1.625cm}\node at (1.25,2*\len) {$=$};
				\node[left] at (0,2*\len) {$\scriptscriptstyle A$};
				\node[left] at (1.75,2*\len) {$\scriptscriptstyle A$};}{xshift=6cm}
			\ezs{\ezs{\mor{t}\step{\counit}}{xscale=.5}\ezs{\counit}{xshift=1.625cm}\node at (1.25,2*\len) {$=$};
				\node[left] at (0,2*\len) {$\scriptscriptstyle A$};
				\node[left] at (1.75,2*\len) {$\scriptscriptstyle A$};}{xshift=10cm}
		\end{tikzpicture}
	\]
	The associated variety is the category of directed graphs, as model morphisms $F$ must satisfy:
	\[\begin{tikzpicture}[scale=.75,baseline=(current bounding box.center)]
			\mor{s}
			\ezs{\mor{F}}{xshift=1cm}
			\ezs{\mor{F}
				\ezs{\mor{s}}{xshift=1cm}
			}{xshift=3cm}
			\node at (2.5,2*\len) {$=$};
			\ezs{
				\mor{t}
				\ezs{\mor{F}}{xshift=1cm}
				\ezs{\mor{F}
					\ezs{\mor{t}}{xshift=1cm}
				}{xshift=3cm}
				\node at (2.5,2*\len) {$=$};
			}{xshift=6cm}
		\end{tikzpicture}\]

\end{example}

\begin{example}[Reflexive Graphs]\label{ex:reflgraphs}
	Extending Example~\ref{ex:dirgraph}, we ask that each vertex has a self-loop:
	\[
		\begin{tikzpicture}
			\mor{\eed}
			\node[left] at (0,2*\len) {$\scriptscriptstyle O$};
			\node[right] at (1,2*\len) {$\scriptscriptstyle A$};
			\ezs{
				\ezs{\mor{\eed}\step{\counit}}{xscale=.5}\ezs{\counit}{xshift=2cm}\node at (1.25,2*\len) {$=$};
				\node[left] at (0,2*\len) {$\scriptscriptstyle O$};
				\node[left] at (2,2*\len) {$\scriptscriptstyle O$};
			}{xshift=2.5cm}
			\ezs{
				\mor{\eed}
				\step{\mor{s}}
				\node[left] at (0,2*\len) {$\scriptscriptstyle O$};
				\node[right] at (1.75,2*\len) {$\scriptscriptstyle O$};
				\ezs{\id
					\node[left] at (0,2*\len) {$\scriptscriptstyle O$};
					\node[right] at (1,2*\len) {$\scriptscriptstyle O$};
				}{xshift=4.25cm}
				\node at (3,2*\len) {$=$};
				\node at (6.5,2*\len) {$=$};
				\ezs{\mor{\eed}
					\step{\mor{t}}
					\node[left] at (0,2*\len) {$\scriptscriptstyle O$};
					\node[right] at (2,2*\len) {$\scriptscriptstyle O$};
				}{xshift=8cm}
			}{xshift=7cm,xscale=.5}
		\end{tikzpicture}
	\]
	then morphisms of models are required to preserve the self-loop, so the associated variety is the category of \emph{reflexive graphs}. Notice that along with Example~\ref{ex:dirgraph}, this could also be presented as a (total) 2-sorted Lawvere theory, since all the operations are total.
\end{example}

\begin{example}[Categories]\label{ex:cats}
	To capture \emph{categories} we extend Example~\ref{ex:reflgraphs} with a composition operator, which is defined when the target of the first arrow matches the source of the second:
	\[
		\begin{tikzpicture}[xscale=.8,baseline=(current bounding box.center)]
			\emult
			\node[left] at (0,\len) {$\scriptscriptstyle A$};
			\node[left] at (0,3*\len) {$\scriptscriptstyle A$};
			\node[right] at (1,2*\len) {$\scriptscriptstyle A$};
			\ezs{
				\umor{t}
				\lmor{s}
				\step{\mult\step\counit}
			}{xshift=3cm,xscale=.75}
			\ezs{\emult\step\counit
			}{xshift=6cm}
			\eql{5.5}
		\end{tikzpicture}
	\]
	and equations insisting composition is associative and unital, with identities given by the self-loops:
	\begin{gather*}
		\begin{tikzpicture}[xscale=.75,baseline=(current bounding box.center)]
			\umult[red!40]
			\lid
			\step{\emult}
			\ezs{\lmult[red!40]
				\uid
				\step{\emult}
			}{xshift=3.5cm,yshift=\len cm}
			\node[left] at (0,2*\len) {$\scriptstyle A$};
			\node[left] at (0,4*\len) {$\scriptstyle A$};
			\node[left] at (0,\len) {$\scriptstyle A$};
			\node[right] at (2,2*\len) {$\scriptstyle A$};
			\node at (2.625,2*\len) {$=$};
			\node[left] at (3.5,\len) {$\scriptstyle A$};
			\node[left] at (3.5,3*\len) {$\scriptstyle A$};
			\node[left] at (3.5,4*\len) {$\scriptstyle A$};
			\node[right] at (5.5,3*\len) {$\scriptstyle A$};
		\end{tikzpicture}
		\hspace{1cm}
		\begin{tikzpicture}[xscale=.5,baseline=(current bounding box.center)]
			\comult
			\step{\umor{s}\lid\step{\umor{\eed}{\lid}\step{\emult}}}
			\ezs{\id}{xshift=6cm}
			\ezs{\comult
				\step{\lmor{t}\uid\step{\lmor{\eed}{\uid}\step{\emult}}}
			}{xshift=9cm}
			\node[left] at (0,2*\len) {$\scriptstyle A$};
			\node[right] at (4,2*\len) {$\scriptstyle A$};
			\node[left] at (6,2*\len) {$\scriptstyle A$};
			\node[right] at (7,2*\len) {$\scriptstyle A$};
			\node[left] at (9,2*\len) {$\scriptstyle A$};
			\node[right] at (13,2*\len) {$\scriptstyle A$};
			\node at (5,2*\len) {$=$};
			\node at (8,2*\len) {$=$};
		\end{tikzpicture}
	\end{gather*}
	Model morphisms are precisely functors. It is worth noting that this involves an inequality:
	\[
		\begin{tikzpicture}[scale=.75]
			\ezs{
				\emult
				\ezs{\mor{F}}{xshift=1cm}
				\ezs{\umor{F}\lmor{F}
					\ezs{\emult}{xshift=1cm}}{xshift=3cm}
				\node at (2.5,2*\len) {$\le$};
			}{xshift=12cm}
		\end{tikzpicture}
	\]
	This states that if $f$ and $g$ are composable then so are $Ff$ and $Fg$, and in particular $F (f \comp g) = Ff \comp Fg$. If this were an equality, it would insist also that if $Ff$ and $Fg$ are composable, then so are $f$ and $g$, which is not always the case. Of course, the associated variety is the category of small categories.
\end{example}
\begin{example}[\fix{Strict} Monoidal Categories]\label{ex:moncats}
	Next, we extend Example~\ref{ex:cats} by asking for a functorial binary operation $\otimes$ on $O$ and $A$ together with a unit constant $\top$ of $O$:
	\begin{gather*}
\begin{tikzpicture}[xscale=.5,baseline=(current bounding box.center)]
\dcouni{\top}
\node[right] at (1,2*\len) {\tiny $O$};
\end{tikzpicture}
:\hspace{1.5em}
\begin{tikzpicture}[xscale=.5,baseline=(current bounding box.center)]
\dcouni{\top}
\step\counit
\ezs{\akasa}{xshift=3cm,xscale=2}
\node at (2.25,2*\len) {$=$};
\end{tikzpicture}
\hspace{3em}
\begin{tikzpicture}[xscale=.5,baseline=(current bounding box.center)]
\tenmult
\node[left] at (0,3*\len) {\tiny $O$};
\node[left] at (0,\len) {\tiny $O$};
\node[right] at (1,2*\len) {\tiny $O$};
\end{tikzpicture}
:\hspace{1.5em}
\begin{tikzpicture}[xscale=.5,baseline=(current bounding box.center)]
\tenmult
\node[left] at (0,3*\len) {\tiny $O$};
\node[left] at (0,\len) {\tiny $O$};
\step{\counit}
\ezs{\ucounit\lcounit
\node[left] at (0,3*\len) {\tiny $O$};
\node[left] at (0,\len) {\tiny $O$};
}{xshift=4.5cm}
\node at (2.75,2*\len) {$=$};
\end{tikzpicture}\\
\begin{tikzpicture}[xscale=.5,baseline=(current bounding box.center)]
\tenmult
\node[left] at (0,3*\len) {\tiny $A$};
\node[left] at (0,\len) {\tiny $A$};
\step{\counit}
\ezs{\ucounit\lcounit
\node[left] at (0,3*\len) {\tiny $A$};
\node[left] at (0,\len) {\tiny $A$};
}{xshift=4.5cm}
\node at (2.75,2*\len) {$=$};
\end{tikzpicture}
\hspace{3em}
\begin{tikzpicture}[xscale=.5,baseline=(current bounding box.center)]
\tenmult
\node[left] at (0,3*\len) {\tiny $A$};
\node[left] at (0,\len) {\tiny $A$};
\node[right] at (2,2*\len) {\tiny $O$};
\step{\mor{s}}
\ezs{
\umor{s}\lmor{s}\step{\tenmult}
\node[left] at (0,3*\len) {\tiny $A$};
\node[left] at (0,\len) {\tiny $A$};
\node[right] at (2,2*\len) {\tiny $O$};
}{xshift=5cm}
\node at (3.5,2*\len) {$=$};
\end{tikzpicture}\\
\hspace{3em}
\begin{tikzpicture}[xscale=.5,baseline=(current bounding box.center)]
\tenmult
\node[left] at (0,3*\len) {\tiny $A$};
\node[left] at (0,\len) {\tiny $A$};
\node[right] at (2,2*\len) {\tiny $O$};
\step{\mor{t}}
\ezs{
\umor{t}\lmor{t}\step{\tenmult}
\node[left] at (0,3*\len) {\tiny $A$};
\node[left] at (0,\len) {\tiny $A$};
\node[right] at (2,2*\len) {\tiny $O$};
}{xshift=5cm}
\node at (3.5,2*\len) {$=$};
\end{tikzpicture}
\hspace{3em}
\begin{tikzpicture}[xscale=.5,baseline=(current bounding box.center)]
\up[2]\tenmult
\down[2]\tenmult
\step{\down[2]{\ezs{\emult}{yscale=2}}}
\ezs{\up[2]\uid\braid\down[2]\lid
\step{\up[2]\emult
\down[2]\emult}
\step[2]{\down[2]{\ezs{\tenmult}{yscale=2}}}
\node[left] at (0,3*\len) {\tiny $A$};
\node[left] at (0,5*\len) {\tiny $A$};
\node[left] at (0,\len) {\tiny $A$};
\node[left] at (0,-\len) {\tiny $A$};
\node[right] at (3,2*\len) {\tiny $A$};
}{xshift=5cm,xscale=1.25}
\node[left] at (0,3*\len) {\tiny $A$};
\node[left] at (0,5*\len) {\tiny $A$};
\node[left] at (0,\len) {\tiny $A$};
\node[left] at (0,-\len) {\tiny $A$};
\node[right] at (2,2*\len) {\tiny $A$};
\node at (3.5,2*\len) {$=$};
\end{tikzpicture}
\end{gather*}
	Additionally, we require equations to the effect that $\otimes$ is associative and unital:
	\begin{gather*}
\begin{tikzpicture}[xscale=.5,baseline=(current bounding box.center)]
\up\tenmult
\lid
\step{\tenmult}
\node[left] at (0,\len) {\tiny $O$};
\node[left] at (0,2*\len) {\tiny $O$};
\node[left] at (0,4*\len) {\tiny $O$};
\node[right] at (2,2*\len) {\tiny $O$};
\up{\ezs{
\up\id
\down\tenmult
\step{\tenmult}
\node[left] at (0,0) {\tiny $O$};
\node[left] at (0,2*\len) {\tiny $O$};
\node[left] at (0,3*\len) {\tiny $O$};
\node[right] at (2,2*\len) {\tiny $O$};}{xshift=5cm}}
\node at (3.5,2*\len) {$=$};
\end{tikzpicture}\hspace{3em}
\begin{tikzpicture}[xscale=.5,baseline=(current bounding box.center)]
\up\tenmult
\lid
\step{\tenmult}
\node[left] at (0,\len) {\tiny $A$};
\node[left] at (0,2*\len) {\tiny $A$};
\node[left] at (0,4*\len) {\tiny $A$};
\node[right] at (2,2*\len) {\tiny $A$};
\up{\ezs{
\up\id
\down\tenmult
\step{\tenmult}
\node[left] at (0,0) {\tiny $A$};
\node[left] at (0,2*\len) {\tiny $A$};
\node[left] at (0,3*\len) {\tiny $A$};
\node[right] at (2,2*\len) {\tiny $A$};}{xshift=5cm}}
\node at (3.5,2*\len) {$=$};
\end{tikzpicture}\\[3mm]
\begin{tikzpicture}[xscale=.5,baseline=(current bounding box.center)]
\up{\dcouni{\top}}
\down{\id}
\step{\up{\mor{\eed}}\down\id}
\step[2]\tenmult
\node[left] at (0,\len) {\tiny $A$};
\node[right] at (3,2*\len) {\tiny $A$};
\node at (4.5,2*\len) {$=$};
\ezs{\down{\dcouni{\top}}
\up{\id}
\step{\down{\mor{\eed}}\up\id}
\step[2]\tenmult
\node[left] at (0,3*\len) {\tiny $A$};
\node[right] at (3,2*\len) {\tiny $A$};
}{xshift=6cm}
\end{tikzpicture}\hspace{3em}
\begin{tikzpicture}[xscale=.5,baseline=(current bounding box.center)]
\up\id
\down{\dcouni{\top}}
\step\tenmult
\node[left] at (0,3*\len) {\tiny $O$};
\node[right] at (2,2*\len) {\tiny $O$};
\ezs{
\id
\node[left] at (0,2*\len) {\tiny $O$};
\node[right] at (1,2*\len) {\tiny $O$};
}{xshift=5cm}
\ezs{
\down\id
\up{\dcouni{\top}}
\step\tenmult
\node[left] at (0,\len) {\tiny $O$};
\node[right] at (2,2*\len) {\tiny $O$};
}{xshift=9cm}
\node at (3.5,2*\len) {$=$};
\node at (7.5,2*\len) {$=$};
\end{tikzpicture}
\end{gather*}
	Now the associated variety is the category of strict monoidal categories and strict monoidal functors.
\end{example}

\begin{example}[Symmetric \fix{Strict} Monoidal Categories]\label{ex:symmoncats}
	To capture \emph{symmetric monoidal categories}, we extend Example~\ref{ex:moncats} with a binary operation $\sigma : O \otimes O \to A$ for the braiding maps, subject to:
	\begin{gather*}
\begin{tikzpicture}[xscale=.75,baseline=(current bounding box.center)]
\step{\cothingy{\sigma}}
\up{\gau{O}}
\down{\gau{O}}
\dro{A}
\end{tikzpicture}
:\hspace{1.5em}
\begin{tikzpicture}[xscale=.75,baseline=(current bounding box.center)]
\cothingy{\sigma}\counit
\ezs{\ucounit\lcounit
\up{\gau{O}}
\down{\gau{O}}
}{xshift=3cm}
\eql{1.5}
\end{tikzpicture}
\hspace{3em}
\begin{tikzpicture}[xscale=.75,baseline=(current bounding box.center)]
\step{\cothingy{\sigma}\mor{s}}
\up{\gau{O}}
\down{\gau{O}}
\step{\dro{O}}
\ezs{\tenmult
\up{\gau{O}}
\down{\gau{O}}
\dro{O}
}{xshift=5cm}
\eql{3.5}
\end{tikzpicture}\\
\begin{tikzpicture}[xscale=.75,baseline=(current bounding box.center)]
\step{\cothingy{\sigma}\mor{t}}
\up{\gau{O}}
\down{\gau{O}}
\step{\dro{O}}
\ezs{\braid\step{\tenmult
\step[-1]{\up{\gau{O}}
\down{\gau{O}}}
\dro{O}}
}{xshift=5cm}
\eql{3.5}%
\end{tikzpicture}\hspace{1em}
\begin{tikzpicture}[xscale=.75,baseline=(current bounding box.center)]
\up[2]\comult
\down[2]\comult
\step{\up[2]\uid\braid\down[2]\lid}
\up[2]{\step[3]{\cothingy{\sigma}}}
\down[2]{\step[2]\braid}
\down[2]{\step[4]{\cothingy{\sigma}}}
\step[3]{\up[2]\id}
\down[2]{\step[4]{\ezs{\emult}{yscale=2}}}
\up[2]{\gau{O}}
\down[2]{\gau{O}}
\step[4]{\dro{A}}
\eql{6}
\ezs{\tenmult\step{\mor{\eed}}
\up[2]{\gau{O}}
\down[2]{\gau{O}}
\step{\dro{A}}
}{xshift=7cm}
\end{tikzpicture}\\
\begin{tikzpicture}[xscale=.75,baseline=(current bounding box.center)]
\up[2]{\comult\step\uid}
\down[2]{\comult\step\lid}
\step\braid
\step[2]{\up[2]\tenmult}
\step[2]{\down[2]{\umor{t}\lmor{t}}}
\step[3]{\up\uid}
\step[4]{\down[2]{\cothingy{\sigma}}}
\step[4]{\ezs{\down\emult}{yscale=2}}
%%%%%%%%%%%
\ezs{\up[2]{\comult\step\uid}
\down[2]{\comult\step\lid}
\step\braid
\step[3]{\up[2]\tenmult}
\step[2]{\down[2]{\umor{s}\lmor{s}}\up[2]\braid}
% \step[3]{\up\uid}
\step[4]{\down[2]{\cothingy{\sigma}}}
\step[4]{\ezs{\down\emult}{yscale=2}}}{yshift=1cm,xshift=7cm,yscale=-1}
\eql{6}
\up[2]{\gau{A}}
\down[2]{\gau{A}}
\step[4]{\dro{A}}
\step[7]{\up[2]{\gau{A}}
\down[2]{\gau{A}}\step[4]{\dro{A}}}
\end{tikzpicture}
\end{gather*}
	This gives the variety of strict monoidal categories and symmetric strict monoidal functors.
\end{example}
\begin{example}[Cartesian Restriction Categories]\label{ex:crcats}
	In light of \ref{thm:newfox}, we can capture \emph{CR categories} by extending Example~\ref{ex:symmoncats} with operations $\delta : O \to A$ and $\varepsilon : O \to A$ corresponding to the comultiplication and counit of the comonoid structure:
	\begin{gather*}
\begin{tikzpicture}[xscale=.75,baseline=(current bounding box.center)]
\mor{\delta} \node[left] at (0,2*\len) {$\scriptstyle O$}; \node[right] at (1,2*\len) {$\scriptstyle A$};
\end{tikzpicture}:
\hspace{1em}
\begin{tikzpicture}[xscale=.75,baseline=(current bounding box.center)]
\counit
\ezs{\mor{\delta}
\step{\counit}}{xshift=2cm}
\node at (1.25,2*\len) {$=$};
\end{tikzpicture}
\hspace{1em}
\begin{tikzpicture}[xscale=.75,baseline=(current bounding box.center)]
\mor{\delta}
\step{\mor{s}}
\ezs{\id}{xshift=3cm}
\node at (2.5,2*\len) {$=$};
\end{tikzpicture}
\hspace{1em}
\begin{tikzpicture}[xscale=.75,baseline=(current bounding box.center)]
\mor{\delta}
\step{\mor{t}}
\ezs{\comult
\step{\tenmult}}{xshift=3cm}
\node at (2.5,2*\len) {$=$};
\end{tikzpicture}\\
\begin{tikzpicture}[xscale=.75,baseline=(current bounding box.center)]
\mor{\varepsilon} \node[left] at (0,2*\len) {$\scriptstyle O$}; \node[right] at (1,2*\len) {$\scriptstyle A$};
\end{tikzpicture}:
\hspace{1em}
\begin{tikzpicture}[xscale=.75,baseline=(current bounding box.center)]
\mor{\varepsilon}
\step{\counit}
\ezs{\counit}{xshift=3cm}
\node at (2.25,2*\len) {$=$};
\end{tikzpicture}
\hspace{1em}
\begin{tikzpicture}[xscale=.75,baseline=(current bounding box.center)]
\mor{\varepsilon}
\step{\mor{s}}
\ezs{\id}{xshift=3cm}
\node at (2.5,2*\len) {$=$};
\end{tikzpicture}
\hspace{1em}
\begin{tikzpicture}[xscale=.75,baseline=(current bounding box.center)]
\mor{\varepsilon}
\step{\mor{t}}
\ezs{\counit
\step{\dcouni{\top}}}{xshift=3cm}
\node at (2.5,2*\len) {$=$};
\end{tikzpicture}
\end{gather*}
	along with equations insisting that $\delta$ and $\varepsilon$ are coherent with respect to the monoidal structure:
	\begin{gather*}
		\begin{tikzpicture}[xscale=.5]
			\tenmult\step{\mor{\delta}}
			\ezs{
				\down\lcomult\up\ucomult
				\step{\up[2]\uid\braid\down[2]\lid}
				\step[2]{\up[3.5]{\ezs{\ucomult}{yscale=.5}}\up[1.5]{\ezs{\ucomult}{yscale=.5}}\down[2]\twoid}
				\step[3]{\lmor{\delta}\down[2]{\lmor{\delta}}\up[.5]\id\up[3.5]{\mor{\eed}}}
				\step[4]{
					\up[3]{\ezs{\cothingy{\sigma}}{yscale=.5}}
				}
				\step[4]{\up[.5]{\mor{\eed}\up[2.75]{\ezs{\tenmult}{yscale=.75}}}\down[2]\tenmult}
				\step[5]{\up[1.75]{
						\draw[wire] (0,\len-.0625) -| (2*\len,3*\len) -- (0,3*\len);
						\draw[wire] (2*\len,2*\len) -- (1,2*\len);
						\node[fill=white,inner sep=-1.25pt,circle] at (2*\len,2*\len) { $\otimes$};}
					\down[2]\id
				}
				\step[6]{\down[1.5]{\ezs{
							\draw[wire] (0,\len-.035) -| (2*\len,3*\len) -- (0,3*\len);
							\draw[wire] (2*\len,2*\len) -- (1,2*\len);
							\node[dot] at (2*\len,2*\len) {};
						}{yscale=1.75}}}
			}{xshift=3cm}
			\node at (2.5,2*\len) {$=$};
		\end{tikzpicture}
		\hspace{2.5cm}
		\begin{tikzpicture}[xscale=.75]
			\tenmult\step{\mor{\varepsilon}}
			\ezs{\umor{\varepsilon}\lmor{\varepsilon}
				\step\tenmult
			}{xshift=3cm}
			\node at (2.5,2*\len) {$=$};
		\end{tikzpicture}
	\end{gather*}
	And finally equations for the commutative comonoid axioms, and naturality of $\delta$:
	\begin{gather*}
\begin{tikzpicture}[xscale=.5,baseline=(current bounding box.center)]
\ezs{\comult}{yscale=2}
\step{\up{\lcomult\up[2]{\umor{\delta}}}}
\step[2]{\umor{\delta}\lmor{\eed}\up[3]\uid}
\step[3]{\tenmult\up[3]\uid}
\step[4]{\ezs{\emult}{yscale=2}}
\ezs{
\ezs{\comult}{yscale=2}
\step{\comult\up[3]{\umor{\delta}}}
\step[2]{\up[3]\uid\umor{\eed}\lmor{\delta}}
\step[3]{\tenmult\up[3]\uid}
\step[4]{\ezs{\emult}{yscale=2}}
}{xshift=7cm}
\node at (6,4*\len) {$=$};
\end{tikzpicture}\hspace{2cm} %break one
\begin{tikzpicture}[xscale=.5,baseline=(current bounding box.center)]
\ezs{\comult}{yscale=2}
\step{\comult\up[3]{\umor{\delta}}}
\step[2]{\up[3]\uid}
\step[3]{\cothingy{\sigma}\ezs{\emult}{yscale=2}}
\ezs{\up[2]{\mor{\delta}}
}{xshift=6cm}
\node at (5,4*\len) {$=$};
\end{tikzpicture}\\ %break two
\begin{tikzpicture}[xscale=.5,baseline=(current bounding box.center)]
\ezs{\comult}{yscale=2}
\step{\up{\lcomult\up[2]{\umor{\delta}}}}
\step[2]{\umor{\varepsilon}\lmor{\eed}\up[3]\uid}
\step[3]{\tenmult\up[3]\uid}
\step[4]{\ezs{\emult}{yscale=2}}
\ezs{
\ezs{\comult}{yscale=2}
\step{\comult\up[3]{\umor{\delta}}}
\step[2]{\up[3]\uid\umor{\eed}\lmor{\varepsilon}}
\step[3]{\tenmult\up[3]\uid}
\step[4]{\ezs{\emult}{yscale=2}}
}{xshift=7cm}
\node at (6,4*\len) {$=$};
\end{tikzpicture}\hspace{1cm} % break 3
\begin{tikzpicture}[xscale=.5,baseline=(current bounding box.center)]
\ezs{\comult}{yscale=2}
\step{\up{\lcomult\up[2]{\umor{s}}}}
\step[2]{\tenmult\up[3]{\umor{\delta}}}
\step[3]{\id\up[3]\uid}
\step[4]{\ezs{\emult}{yscale=2}}
\ezs{
\up[2]{\comult
\step{\lmor{t}
\uid}
\step[2]{\lmor{\delta}\uid}
\step[3]{\emult}
}
}{xshift=7cm}
\node at (6,4*\len) {$=$};
\end{tikzpicture}
\end{gather*}
	The associated variety is the category of CR categories and CR functors.
\end{example}
\begin{example}[Discrete Cartesian Restriction Categories]\label{ex:dcrcats}
	\ref{thm:dcrc} makes it easy to capture \emph{DCR categories} by extending Example~\ref{ex:crcats} with $\mu : O \to A$ satisfying the Frobenius and special equations: there is a \begin{tikzpicture}[xscale=.5,baseline=(current bounding box.center)]
		\premor{\mu}{O}{A}
		\end{tikzpicture} such that
	\begin{gather*}
%\begin{tikzpicture}[xscale=.5,baseline=(current bounding box.center)]
%\tmor{\mu}{O}{A}
%\end{tikzpicture}
%:\hspace{1.5em}
\begin{tikzpicture}[xscale=.5,baseline=(current bounding box.center)]
\mor{\mu}\step{\counit}
\gau{O}
\ezs{\counit\gau{O}}{xshift=4cm}
\eql{2.5}
\end{tikzpicture}
\hspace{9em}
\begin{tikzpicture}[xscale=.5,baseline=(current bounding box.center)]
\mor{\mu}
\step{\mor{s}}
\gau{O}\step{\dro{O}}
\ezs{
\comult\step{\tenmult}
\gau{O}\step{\dro{O}}
}{xshift=5cm}
\eql{3.5}
\end{tikzpicture}\\
\begin{tikzpicture}[xscale=.5,baseline=(current bounding box.center)]
\mor{\mu}
\step{\mor{t}}
\gau{O}\step{\dro{O}}
\ezs{
\id\gau{O}\dro{O}
}{xshift=5cm}
\eql{3.5}
\end{tikzpicture}
\hspace{6.5em}
\begin{tikzpicture}[xscale=.5,baseline=(current bounding box.center)]
\comult
\step{\umor{\delta}\lmor{\mu}}
\step[2]{\emult}
\gau{O}\step[2]{\dro{A}}
\eql{4.5}
\ezs{\tmor{\eed}{O}{A}}{xshift=6cm}
\end{tikzpicture}\\
\begin{tikzpicture}[xscale=.5,baseline=(current bounding box.center)]
\ezs{\comult}{yscale=2}
\step{\up[4]\comult\comult}
\step[2]{\umor{\mu}\lmor{\eed}\up[4]{\umor{\eed}\lmor{\delta}}}
\step[3]{\tenmult\up[4]\tenmult}
\step[4]{\ezs{\emult}{yscale=2}}
\up[2]{\gau{O}\step[4]{\dro{A}}}
%%%%%%%%%%%%%%%%%
\ezs{
\comult
\step{\lmor{\delta}\umor{\mu}}
\step[2]{\emult}
\gau{O}\step[2]{\dro{A}}
}{yshift=2*\len cm, xshift=8cm}
%%%%%%%%%%%%%%%%%
\ezs{\ezs{\comult}{yscale=2}
\step{\up[4]\comult\comult}
\step[2]{\umor{\eed}\lmor{\mu}\up[4]{\umor{\delta}\lmor{\eed}}}
\step[3]{\tenmult\up[4]\tenmult}
\step[4]{\ezs{\emult}{yscale=2}}
\up[2]{\gau{O}\step[4]{\dro{A}}}}{xshift=14cm}
\up[2]{\eql{6.5}\eql{12.5}}
\end{tikzpicture}
\end{gather*}
	The variety is the category of strict DRC categories and strict CR functors (since they preserve $\mu$).
\end{example}
\begin{example}[Cartesian Categories]\label{ex:cartcats}
	To capture \emph{cartesian categories} instead, we can extend Example~\ref{ex:crcats} with one equation, ensuring that $\varepsilon$ is natural:
	\[
		\begin{tikzpicture}[xscale=.5,baseline=(current bounding box.center)]
			\comult
			\step{\uid\lmor{t}}
			\step[2]{\uid\lmor{\varepsilon}}
			\step[3]\emult
			\gau{A}
			\step[3]{\dro{A}}
			\ezs{\mor{s}\step{\mor{\varepsilon}}
				\gau{A}\step{\dro{A}}}{xshift=7cm}
			\eql{5.5}
		\end{tikzpicture}
	\]
	Then by Theorem~\ref{thm:fox}, this gives the variety of strict cartesian categories and strict cartesian functors.
\end{example}
\begin{example}[Cartesian Closed Categories]\label{carclocat}
	Finally, to capture \emph{cartesian closed categories} we extend Example~\ref{ex:cartcats} with an operator $\mathsf{exp} : O \otimes O \to O$, the idea being that $\mathsf{exp}(A,B)$ is the internal hom $[A,B]$, along with an operator $\mathsf{ev} : O \otimes O \to O$ that gives the corresponding evaluation map:
	\begin{gather*}
		\begin{tikzpicture}[xscale=.75,baseline=(current bounding box.center)]
			\cothingy{\exp}
			\step[-1]{\up{\gau{O}}}
			\step[-1]{\down{\gau{O}}}
			\step[-1]{\dro{O}}
		\end{tikzpicture}:\hspace{1.5em}
		\begin{tikzpicture}[xscale=.75,baseline=(current bounding box.center)]
			\cothingy{\exp}\counit
			\ezs{\ucounit\lcounit}{xshift=2cm}
			\node at (1.25,2*\len) {$=$};
		\end{tikzpicture}
		\hspace{1cm}
		\begin{tikzpicture}[xscale=.75,baseline=(current bounding box.center)]
			\step{\cothingy{\ev}}
			\up{\gau{O}}
			\down{\gau{O}}
			\dro{A}%here
		\end{tikzpicture}:\hspace{1.5em}
		\begin{tikzpicture}[xscale=.75,baseline=(current bounding box.center)]
			\cothingy{\ev}\counit
			\ezs{\ucounit\lcounit}{xshift=2cm}
			\node at (1.25,2*\len) {$=$};
		\end{tikzpicture}\\
		\begin{tikzpicture}[xscale=.75,baseline=(current bounding box.center)]
			\cothingy{\ev}\mor{s}
			\ezs{\ucomult\down\lid
				\step{\up\uid\down\braid}
				\step[2]{\down[2]\id}
				\step[3]{\up{\cothingy{\exp}}}
				\step[3]{\down[1.5]{\ezs{\tenmult}{yscale=1.5}}}
			}{xshift=2cm}
			\node at (1.5,2*\len) {$=$};
		\end{tikzpicture}
		\hspace{1.5cm}
		\begin{tikzpicture}[xscale=.75,baseline=(current bounding box.center)]
			\cothingy{\ev}\mor{t}
			\ezs{
				\ucounit
				\lid
			}{xshift=2cm}
			\node at (1.5,2*\len) {$=$};
		\end{tikzpicture}\hspace{3em}
	\end{gather*}
	along with an operation $\lambda$ and equations stating, intuitively, that $\lambda(X,A,B,f)$ is defined precisely in case $f : X \times A \to B$, and yields a map $\lambda (X,A,B,f) : X \to [A,B]$ as in:
	\begin{gather*}
\begin{tikzpicture}[yscale=.5,baseline=(current bounding box.center)]
\multi{\lambda}
\up[2]{\dro{A}}
\up{\gau{O}}
\up[3]{\gau{O}}
\up[5]{\gau{O}}
\down{\gau{A}}
\end{tikzpicture}
\hspace{1.5em}
\begin{tikzpicture}[xscale=.75,baseline=(current bounding box.center)]
\ezs{
\multi{\lambda}}{yscale=.75,xscale=1.25}
\step{\up\counit}
\ezs{\up[2]\tenmult\down[2]{\ezs{\down\coturn}{yscale=1}\turn}
\step{\up\uid\down{\comult}\down[4]{\ezs{\id}{xscale=2}}}
\step[2]{\up\uid\mor{s}\down[2]{\mor{t}}}
\step[3]{\up\mult\down[3]\mult}
\step[4]{\up\counit\down[3]\counit}
}{xshift=3cm,yshift=.5cm}
\up{\eql{2.25}}
\end{tikzpicture}
% \end{document}
\\
\begin{tikzpicture}[baseline=(current bounding box.center),yscale=.75]
\multi{\lambda}
\step{\up[2]{\mor{s}}}
\ezs{\up[5]\id
\up[3]\counit
\up\counit
\down\counit
}{xshift=3cm}
\up[2]{\eql{2.5}}
\end{tikzpicture}\hspace{4em}
\begin{tikzpicture}[baseline=(current bounding box.center),yscale=.75]
\multi{\lambda}
\step{\up[2]{\mor{t}}}
\ezs{\up[5]\counit
\up[2]{\step{\cothingy{\exp}}}
\down\counit
}{xshift=3cm}
\up[2]{\eql{2.5}}
\end{tikzpicture}
\end{gather*}
	also equations insisting that if $f : X \times A \to B$ then  $(\lambda(X,A,B,f) \times 1)\comp \mathsf{ev} = f$ holds:
	\[
		\begin{tikzpicture}[xscale=.75,baseline=(current bounding box.center)]
\up{\uid}
\lcomult
\down[6]{\ucomult\down{\lid}}
\step{\up{\twoid}\down[3]{\braid}\down[7]{\braid}\step{\ezs{\down[7]\lid}{xscale=2}\ezs{\up\uid\down\twoid}{xscale=2}\down[5]\braid\step{\down[6]\comult\down[4]\id\step{\down[5]{\mor{\eed}}}}}}
\down[3]{\ezs{\multi{\lambda}\ezs{\cothingy{\ev}}{yshift=-3.5*\len cm, xshift=1cm,yscale=.5}\step{\ezs{\down\tenmult}{yscale=2}}}{xshift=4cm}}
\ezs{\down[6.5]{
\draw[wire] (0,\len) -| ++(.625,2.5*\len) -- ++(2*\len,0);}
}{xshift=5cm}
\ezs{\down[4]\emult}{xshift=6cm}
\end{tikzpicture}\hspace{1.5em}
=
\hspace{1.5em}
\begin{tikzpicture}[xscale=.75,baseline=(current bounding box.center)]
\up{\tenmult}\down[5]\lid\ezs{\down[3]\id}{xscale=2}
\step{\ezs{\id}{xscale=4,yshift=\len cm}\down[4]{\down\lcomult}}
\step[2]{\down[4]\coturn\ezs{\down[6]\lid}{xscale=5}\ezs{\down[3]\turn}{yscale=1.5}}
\step[3]{\down[5]{\ezs{\id}{xscale=2}}\down[2]\comult}
\step[4]{\lmor{s}\down[2]{\lmor{t}}}
\step[5]{\ezs{\mult}{yshift=.5*\len}\down[4.5]{\ezs{\mult}{yshift=.5*\len cm}}}
\step[6]{\down[.5]{\up[.5]\counit\down[3.5]\counit}}
\end{tikzpicture}
	\]
	and that if $g : X \to [A,B]$ then $\lambda(X,A,B,(g \times 1)\comp \mathsf{ev}) = g$ holds:
	\[
		\begin{tikzpicture}[xscale=.75,baseline=(current bounding box.center)]
\ezs{\uid}{xscale=4}\down[6]\comult\step{\down[2]{\cothingy{\rotatebox[origin=c]{90}{$\scriptscriptstyle\exp$}}}}
\step{\down[4.5]{\ezs{\braid}{yscale=1.5}}\ezs{\down[6]\lid}{xscale=5}}
\step[2]{\down[2]\comult\ezs{\down[5]\id}{xscale=2}}
\step[3]{\lmor{s}\down[2]{\lmor{t}}}
\step[4]{\mult\down[4.5]{\ezs{\mult}{yshift=.5*\len cm}}}
\step[5]{\counit\down[4]\counit}
\end{tikzpicture}
\hspace{1.5em}=\hspace{1.5em}
%%%%%%%%%%%%%%%%%%%%%%%%%%%%%%%%%%%%%%%%%%%%%%
\begin{tikzpicture}[xscale=.75,baseline=(current bounding box.center)]
\up[4]{\ezs{\id}{xscale=6}}\ezs{\lcomult}{yscale=2}\down[5]\twoid
\step[1]{\up[2]\comult\down[3]{\braid}\down[5]\lid}
\step[2]{\down[2]\comult\down[5]\braid\up[2]\twoid}
\step[3]{\ezs{\up[3]\id}{xscale=3}\down[2]\lid\down[3]{\lid\down[2]{\lmor{\eed}}}}
\step[4]{\ezs{\down[1.675]\braid}{yscale=1.5}\cothingy{\ev}\down[5]\tenmult}
\step[5]{\down[4]\braid\down\uid}
\step[6]{\down[4]\emult\down\coturn\up[2]\coturn\up[3]\coturn}
\step[7]{\down\id}
\step[6.75]{
\ezs{\draw[fill=white, thin] (\len,\len-.15) rectangle (3*\len,3*\len+.15) node[pos=.5,] {$\scriptscriptstyle \lambda$};}{yshift=-.875cm,yscale=2.25}}
\end{tikzpicture}
	\]
	Now the associated variety is the category of strict cartesian closed categories and strict cartesian closed functors: these preserve hom-objects and, when $\lambda(X,A,B,f)$ is defined, satisfy $F\lambda(X,A,B,f) = \lambda(FX,FA,FB,Ff)$.
	This presentation of cartesian closed categories is essentially due to Freyd: a version of it is given immediately after the first appearance of the notion of essentially algebraic theory in \cite{freyd1972aspects}, albeit somewhat informally, and using very different syntax.
\end{example}

\section{The Variety Theorem for Partial Theories}\label{the_variety}
Here we classify the categories of models of partial Lawvere theories. These turn out to be exactly the locally finitely presentable (LFP) categories~\cite[1A]{adamek_rosicky_1994}.
LFP categories have an important position in categorical algebra, due to deep connections with model theory \cite[Ch. 5]{adamek_rosicky_1994} and \cite{makkai1989accessible}, homotopy theory \cite{Dugger2001}, and universal algebra \cite[Ch. 3]{adamek_rosicky_1994}.
%
% We recall some of the definitions in Appendix B; a more complete account is in the aforementioned references.
\begin{definition}
	In a category $\clC$, an object $C$ is \emph{finitely presentable} if the hom-functor $\clC(C,\firstblank)$ preserves directed colimits (see \cite[1.1]{adamek_rosicky_1994} for the definition). %A \emph{generator} is an object $C$ of $\bfC$ s.t.\ the hom-functor $\bfC(C,\firstblank)$ is faithful. This means, it ``distinguishes arrows'' in that if $\bfC(C,f) = \bfC(C,g)$ for two parallel arrows $f,g$, then $f=g$.
\end{definition}
This notion might appear obscure to the reader unfamiliar with categorical logic; \cite[1.2]{adamek_rosicky_1994} contains lots of examples to help the reader build their intuition: for instance, an object of the category of sets is finitely presentable if and only if it is finite, and a (commutative) monoid is finitely presentable if and only if it admits a presentation $\langle G\mid R\rangle$ where both $G$ (set of generators) and $R$ (set of relations) are finite sets: this happens for many other algebraic structures, and thus motivates the definition.
% A generator $C$ of $\bfC$ is \emph{strong} if in addition for each $m : A \hookrightarrow X$ that is not an isomorphism, there exists a monomorphism $C\to X$ that does not factor through $m$.

% More generally, a set of objects $\{C_i\mid i : I\}$ of a category is called a (strong) generator if the entire family distinguishes objects (and monics): if for all $i: I$ one has $\bfC(C_i,f) = \bfC(C_i,g)$ for two parallel arrows $f,g$, then $f=g$, and for each $m : A \hookrightarrow X$ that is not an isomorphism, there exist an $i :I$ and a monomorphism $C_i\to X$ that does not factor through $m$.
\begin{definition}[Locally finitely presentable category]\label{lfp_cat}
	\cite[Def. 1.9]{adamek_rosicky_1994} A locally finitely presentable (LFP) category $\clK$ is a cocomplete category s.t.\ there is a small full subcategory $\clA\subset\clK$ of finitely presentable objects, and such that every object of $\clK$ is a directed colimit of objects of $\clA$.% is a strong generator for $\clK$.
\end{definition}
%LFP categories abound in the mathematical practice;
% \begin{example}
% 	\begin{enumtag}{lp}
% 		\item $\clK=\ct{Set}$, the category of sets. There, the subcategory of finitely presentable objects is made by finite sets $\{[0],[1],[2],\dots\}$.
% 		\item $\ct{cMon}$, the category of commutative monoid. There, the subcategory of finitely presentable objects is made by finitely generated commutative monoids.\footnote{A minor technical point, here and everywhere else: strictly speaking the subcategories of finitely presentable objects are \emph{essentially} small, meaning that they have small skeleton; the difference between smallness and essential smallness can be safely ignored.}
% 		\item $\clK=\ct{Gra}$, the category of graphs. There, the subcategory of finitely presentable objects is made by finitely generated graphs (a graph is finitely generated if it has a finite number of edges and vertices?).
% 		\item $\clK=\ct{Pos}$, the category of posets. There, the set of finitely presentable objects is made by posets whose underlying set is finite.
% 	\end{enumtag}
% \end{example}
%the reader must however be warned that familiar examples of categories, like topological spaces, are not locally presentable.%: for example, the category of topological spaces is not an LFP category.

As in the classical case (Remark~\ref{rem:sorts}), the most crisp statement of the variety theorem is for the unsorted case. Just as an unsorted Lawvere theory is exactly a (small) cartesian category, we define an \emph{unsorted partial Lawvere theory} to be a (small) DCR category, and the corresponding notion of morphism to be a CR functor. Then:
\begin{center}
	\fbox{Categories of models of partial theories are exactly LFP categories.}
\end{center}
Indeed, we have a similar contravariant adjunction to that of Theorem~\ref{thm:lawvereadjunction}, if $\mathsf{LFP}$ is a 2-category having 1-cells $R:\clK \to \clK'$ the right adjoint functors $R$ preserving directed colimits, and 2-cells all natural transformations $\alpha : R \To R'$. A motivation for this apparently peculiar choice of 1-cells can be found in our Observation \ref{obs:modeladjunction}; it is exactly as our Definition \ref{morovar}, provided one replaces ``sifted'' colimit with ``directed''.
\begin{theorem}\label{thm:main}
	There is a 2-adjunction
	\begin{equation}\label{the_variety_theorem}
		\Th : \mathsf{LFP} \leftrightarrows (\mathsf{DCRC}^\le)^\op : \Mod ,\end{equation}
	where $\ct{DCRC}^\leq$ is the 2-category of DCR categories of our Definition \ref{the_def_of_dcrc} and $\mathsf{LFP}$ is the 2-category of LFP categories. Moreover, the unit
	of this adjunction is an equivalence, namely there is a natural equivalence of categories between $\clK \in \mathsf{LFP}$ and $\Mod(\Th(\clK))$,
	i.e. each LFP category is equivalent to the category of models of its induced theory.
\end{theorem}
The proof of Theorem~\ref{thm:main} is split into two parts, as illustrated below:
\begin{equation}\label{eq:proofroadmap}
	\begin{tikzcd}
		(\mathsf{DCRC}^\leq)^\text{op} \arrow[rr, shift right] \arrow[rd, dashed, shift right=2]
		& & (\Lex)^\text{op} \arrow[ld, shift left] \arrow[ll, "\fbox{1}" description, shift right] \\
		& \mathsf{LFP} \arrow[ru, "\fbox{2}" description, shift left] \arrow[lu, dashed] &
	\end{tikzcd}
\end{equation}
\begin{enumerate}
	\item we show that $\mathsf{Lex}$ --the 2-category of categories $\clA$ with finite limits, functors $\clA \to \clA'$ preserving finite limits, and natural transformations-- is reflective in the 2-category $\mathsf{DCRC}^\leq$.
	      This is the original, technical core of Theorem \ref{thm:main}.
	\item we connect $\mathsf{Lex}^\op$ and $\mathsf{LFP}$ with a contravariant biequivalence of 2-categories. This is a classical result called \emph{Gabriel-Ulmer duality}.
\end{enumerate}

Composing the two, we obtain Theorem \ref{thm:main}.

We will start from the first of the two tasks, providing an adjunction of 2-categories as follows.
\[  K_t : \mathsf{DCRC}^\leq \leftrightarrows \Lex : \Par \]

We first describe the left adjoint $K_t$, %in \S\ref{sec:splitting},
then the right adjoint $\Par$,  %in \S\ref{sec:par}
and conclude by showing that they define an adjunction.
%\S\ref{sec:mainadjunction}.

\paragraph{Splitting Domain Idempotents}\label{sec:splitting}
The functor $K_t$ arises via a modified Karoubi envelope, also called \emph{Cauchy completion} in \cite{borcauchy}.
Recall that an idempotent $a : A \to A$ in a category \emph{splits} if there is a commutative diagram
\[
	\begin{tikzcd}
		& X \ar[r,"s"] & A \\
		X \ar[r,"s"'] \ar[ur,"1_X"] & A \ar[u,"r"] \ar[ur,"a"']
	\end{tikzcd}
\]
Restriction categories in which all of the domain idempotents split are called \emph{split} restriction categories. An example is $\mathsf{Par}(\bbC)$: for any arrow $(m,f) : A \to B$, $\rest{(m,f)} = (m,m) : A \to A$ splits with $s = (1,m)$ and $r = (m,1)$. Notice that this means the domain of definition of $(m,f)$ is a subobject of $A$. This is a good way to think of split domain idempotents in general: for $\rest{f}$ to be split in a restriction category is for the domain of definition of $f : A \to B$ to be a subobject of $A$.

For any restriction category $\bbX$ we can construct a split restriction category $K(\bbX)$ that contains $\bbX$ as a subcategory.
Its subcategory of \emph{total} maps $K_t(\bbX)$ is of particular interest.
\begin{definition}
	Let $\bbX$ be a DCRC. Then $K_t(\bbX)$ is the category where
	\begin{enumerate}[]
		\item objects are pairs $(A,a)$ with $A$ an object of $\bbX$ and $a : A \to A$ a domain idempotent in $\bbX$;
		\item arrows $f : (A,a) \to (B,b)$ are arrows $f : A \to B$ of $\bbX$ such that $\rest{f} = a$ and $f\comp  b = f$;
		\item composition is given by composition in $\bbX$;
		\item The identity on $(A,a)$ is given by $a$.
	\end{enumerate}
\end{definition}

\noindent It is routine to verify that this forms a category.
Crucially, if $\bbX$ is a DCRC, then the subcategory $K_t(\bbX)$ of \emph{total} maps of $K(\bbX)$ has finite
limits~\cite{Coc12}:

\begin{lemma}
	For any DCRC $\bbX$, $K_t(\bbX)$ has finite limits.
\end{lemma}

We now show that this extends to a 2-functor $K_t : \mathsf{DCRC}^\leq \to \mathsf{Lex}$. If $\bbX$ and $\bbY$ are DCRCs and $F : \bbX \to \bbY$ is a CR functor, then there is a functor $K_t(F) : K_t(\bbX) \to K_t(\bbY)$ defined by $K_t(F)(A,a) = (FA,Fa)$ on objects and $K_t(F)(f) = f$ on arrows. It follows from our characterization of CR functors in terms of the partial Frobenius algebra structure that $K_t(F)$ preserves finite limits, giving the action of our 2-functor $K_t$ on 1-cells. The action of $K_t$ on 2-cells is given as follows:
\begin{lemma}
	If $F,G : \bbX \to \bbY$ are CR functors between DCR categories and $\alpha : F \to G$ is a lax transformation, define $K_t(\alpha) : K_t(F) \to K_t(G)$ by letting $K_t(\alpha)$ at $(A,a)$ in $K_t(\bbX)$ be:
	\[ K_t(\alpha)_{(A,a)} = Fa\comp \alpha_A : (FA,Fa) \to (GA,Ga) \]
	Then $K_t(\alpha)$ is a natural transformation.
\end{lemma}

At this point we need only show that $K_t$ preserves composition and identities for 1-cells and 2-cells, which in both cases is straightforward.% We have now shown:
\begin{lemma}
	$K_t : \mathsf{DCRC}^\leq \to \Lex$ is a 2-functor.
\end{lemma}

%\subsection{Partial Functions Revisited}\label{sec:par}
\paragraph{Partial Functions Revisited}%\label{sec:par}

Here we show that the $\mathsf{Par}$ construction (\S\ref{sec:pardef}) also extends to a 2-functor $\mathsf{Par} : \mathsf{Lex} \to \mathsf{DCRC}^\leq$. If $\mybb{C}$ and $\mybb{D}$ are categories with finite limits and $F : \mybb{C} \to \mybb{D}$ is a finite-limit preserving functor, then we obtain a CR functor $\Par(F) : \Par(\mybb{C}) \to \Par(\mybb{D})$, defined on objects by $\Par(F)(A) = F(A)$, and on arrows by
\[
	\begin{tikzcd}
		& X \ar[ld,"m"'] \ar[rd,"f"]\\
		A && B
	\end{tikzcd}
	\stackrel{\Par(F)}{\longmapsto}
	\begin{tikzcd}
		& FX \ar[ld,"Fm"'] \ar[rd,"Ff"]\\
		FA && FB
	\end{tikzcd}
\]
Since $F$ preserves finite limits, we have that $\Par(F)(\delta_A) = (F1_A,F\Delta_A) = (1_{FA},\Delta_{FA}) = \delta_{FA} = \delta_{\Par(F)(A)}$ and $\Par(F)(\varepsilon_A) = (F1_A, F!_A) = (1_{FA},!_{FA}) = \varepsilon_{\Par(F)(A)}$, so $\Par(F)$ preserves the CR structure. This defines the action of $\mathsf{Par}$ on 1-cells. We present the action of $\mathsf{Par}$ on 2-cells as a lemma:

\begin{lemma}
	If $F,G : \mybb{C} \to \mybb{D}$ are finite limit preserving functors between categories with finite limits and $\alpha : F \to G$ is a natural transformation, define $\Par(\alpha) : \Par(F) \to \Par(G)$ by defining the component of $\Par(\alpha)$ at $A$ in $\mybb{C}$ to be:
	\[
		\scriptsize\begin{tikzcd}
			& FA \ar[d,phantom,"\tiny \Par(\alpha)_A"] \ar[ld,"1_{FA}"'] \ar[dr,"\alpha_A"] \\
			FA &\text{}& GA
		\end{tikzcd}
	\]
	Then $\Par(\alpha) : \Par(F) \to \Par(G)$ is a lax transformation.
\end{lemma}

It remains only show that $\Par$ preserves composition and identities at the level of 1-cells and 2-cells, which is immediate in both cases. We therefore have:

\begin{lemma}
	$\Par : \Lex \to \mathsf{DCRC}^\leq$ is a 2-functor.
\end{lemma}

\paragraph{Adjointness}%\label{sec:mainadjunction}
The following result is original, and builds on~\cite[Corollary 3.5]{Coc02}; however, there the 2-cells of the categories involved are different.
\begin{theorem} \label{LexDCRC}
	There is a 2-adjunction $K_t : \mathsf{DCRC}^\leq \leftrightarrows \Lex : \Par$.
\end{theorem}

It is worth describing the unit and counit of our adjunction. The unit $\eta : 1 \to \Par\circ K_t $ is given by the canonical inclusion $\eta_\bbX : \bbX \to \Par(K_t(\bbX))$ defined by
\[
	A \stackrel{f}{\to} B \hspace{1cm} \stackrel{\eta_\bbX}{\longmapsto}
	\begin{tikzcd}
		& (A,\rest{f}) \ar[ld,"\rest{f}"'] \ar[rd,"f"] \\
		(A,1_A) && (B,1_B)
	\end{tikzcd}
\]
The counit $\varepsilon : K_t\circ \Par  \to 1$ is defined in terms of the equivalence of categories $K(\bbX) \simeq \bbX$ between any split restriction category $\bbX$ and the result of formally splitting its domain idempotents. In particular, since $\Par(\bbC)$ is always split, we obtain an equivalence of categories $K(\Par(\bbC)) \simeq \Par(\bbC)$. Restricting this to the subcategories of total maps gives defines our counit $\varepsilon_\bbC : K_t(\Par(\bbC)) \simeq \bbC$. In particular, the fact that the counit is a natural equivalence gives:
\begin{lemma}
	$\mathsf{Lex}$ is a reflective (2-)subcategory of $\mathsf{DCRC}^\leq$.
\end{lemma}

\paragraph{Gabriel-Ulmer duality}
To complete the triangle~\eqref{eq:proofroadmap}, we recall a theorem first shown by P. Gabriel and F. Ulmer~\cite{gabriel2006lokal}, establishing a contravariant equivalence between the 2-category $\ssLFP$ of locally finitely presentable categories and the 2-category $\mathsf{Lex}$ of categories with finite limits.
%The choice of 1- and 2-cells in $\ssLFP$ and $\mathsf{Lex}$ is made precise right after the statement of the theorem in \ref{GU} below.

The duality asserts that a locally finitely presentable category $\clK$ can be reconstructed from its subcategory $\clK_{\omega}$ of finitely presentable objects. %(see \ref{lfp_cat}). %The proof relies on basic notions and techniques typical of the theory of LFP categories; all that follows can be considered standard; however, we invite the reader to consult our appendix for a terse account of the proof.
A good reference for the proof is \cite[Th. 3.1]{GUcentazzo}.

\begin{theorem}[Gabriel-Ulmer duality] \label{GU}
	There is a biequivalence of 2-categories \[\Lex^\op \leftrightarrows \ssLFP\]
	between $\Lex$, the 2-category of small categories with finite limits, where 1-cells are functors preserving finite limits and 2-cells are the natural transformations, and $\ssLFP$, the 2-category of locally finitely presentable categories, where 1-cells are right adjoints preserving directed colimits.
\end{theorem}
%\begin{proof}
%	We sketch the proof in \ref{}
%\end{proof}
%, a more recent one that sets the result in the enriched context is \cite{GUlack}; one that studies Gabriel-Ulmer duality from the point of view of formal category theory is \cite[Th. 4.12]{GUdll}. We address the reader to \ref{gu_proof} for an outline of the argument.

\subsection{Sorted Gabriel-Ulmer duality}\label{gusection}
%\subsubsection{Sorted Gabriel Ulmer duality}
A similar version of the above theorem holds if, instead of considering theories of all possible sorts, we fix once and for all a single cardinality for the sorts $\sSigma$. Such ``relative'' version of Gabriel-Ulmer duality is useful to recover the classical Lawvere-style approach of single- and many-sorted theories.
\begin{definition}\label{sortedgu1}
	We call $\L\sSigma$ the free category with finite limits over the discrete set $\sSigma$. When $\sSigma$ is the singleton we will use the shortened notation $\L 1$.
\end{definition}

\begin{definition}\label{sortedgu2}
	A $\sSigma$-sorted category with finite limits $(\clA, p)$ is an object in $(\Lex)^\op{/\L\sSigma}$ whose specifying functor $p: \L\sSigma \to \clA$ is bijective on objects. $(\sSigma\text{-}\Lex)^\op$ is the full 2-subcategory of $\sSigma$-sorted categories with finite limits.
\end{definition}

\begin{definition}\label{sortedgu3}
	A $\sSigma$-sorted locally finitely presentable category $(\clK, U)$ is an object in $\ssLFP{/[\sSigma,\Set]}$ whose specifying functor $U: \clK \to [\sSigma,\Set]$ is conservative. $(\sSigma$-$\ssLFP)^\op$ is the full 2-subcategory of $\sSigma$-sorted locally finitely presentable categories.
\end{definition}

\begin{theorem}[Sorted Gabriel-Ulmer duality]\label{sortedGU}
	There is a biequivalence of 2-categories
	\[
		\ct{Mod}_{\sSigma}: (\sSigma\text{-}\Lex)^\op \leftrightarrows \sSigma\text{-}\ssLFP :\Th_\sSigma.
	\]
\end{theorem}
\subsubsection{Sorted partial variety theorem}
We can use the sorted version of Gabriel-Ulmer duality to infer the sorted version of the syntax-semantics duality for multi-sorted partial Lawvere theories.
\begin{theorem}
	There is an 2-adjunction, whose unit is an equivalence, %between 
	\[ \sSigma \text{-}\mathsf{LFP} \leftrightarrows (\sSigma \text{-}\mathsf{pLaw})^\op ,\]
	where $\sSigma \text{-}\mathsf{pLaw}$ is the 2-category of ``$\sSigma$-sorted partial Lawvere theories'', understood as the analogue of Remark \ref{rem:sorts} for partial theories (see Definition \ref{partial_law}), and $\sSigma \text{-}\mathsf{LFP}$ is the 2-category of $\sSigma$-sorted locally finitely presentable categories.
\end{theorem}
\begin{proof}[Sketch of proof]
	The proof is divided into intermediate steps: each tag on the following two diagrams indicates the section where the proof of the adjunction, or equivalence, is given.
	\setlength{\fboxsep}{1pt}
	\begin{center}
	\begin{adjustbox}{max width=0.925\textwidth}
		\parbox{\linewidth}{\[\begin{tikzcd}[ampersand replacement=\&]
			(\sSigma \text{-}\mathsf{pLaw})^\op \arrow[rr, shift right] \arrow[rd, dashed, shift right] \& \& (\sSigma \text{-}\Lex)^\op \arrow[ll, "\scriptsize\fbox{$\star$}" description, shift right] \arrow[ld, shift left] \& \& (\mathsf{DCRC}^\leq)^\op \arrow[rr, shift right] \arrow[rd, dashed, shift right=2] \& \& (\Lex)^\op \arrow[ld, shift left] \arrow[ll, "\scriptsize\fbox{\ref{LexDCRC}}" description, shift right] \\
			\& \sSigma \text{-}\mathsf{LFP} \arrow[ru, "\scriptsize \fbox{\ref{sortedGU}}" description, shift left] \arrow[lu, dashed, shift right] \& \& \& \& \mathsf{LFP} \arrow[ru, "\scriptsize\fbox{\ref{GU}}" description, shift left] \arrow[lu, dashed] \&
		\end{tikzcd}
	\]
	}
\end{adjustbox}
\end{center}
\setlength{\fboxsep}{2pt}
	The claim in $(\star)$ is the only one that needs to be proven. Yet it is also the most trivial one. We will deduce it directly from \ref{LexDCRC}. Indeed if $\mathsf{Lex}$ is reflective in $\mathsf{DCRC}^\leq$, $(\Lex)^\op /\L\sSigma $ is coreflective in $(\mathsf{DCRC}^\leq)^\op /\Par(\L\sSigma) $, now observe that $\Par(\L\sSigma)$ is precisely the free discrete cartesian restriction category over $\sSigma$. The desired result follows passing to functors bijective on objects in the slice.\qedhere
\end{proof}

\begin{observation}\label{obs:freemodelslfp}
	In analogy with~\ref{obs:freemodels}, we can show that sorted partial Lawvere theories have free models. For the single-sorted case, let $p: \parcat{\bbF^\op} \to \clL$ be a partial Lawvere theory. Indeed we can look at it as a morphism in $\mathsf{DCRC}^\leq$, then the previous theorem produces an adjunction $F \dashv \Mod_{p}$ \[F: \Mod_{\parcat{\bbF^\op}} \leftrightarrows \Mod_{\clL} : \Mod_{p}.\] The functor $\Mod_{p}$ coincides with the \textit{forgetful functor}. Its left adjoint $F$ provides free objects.
\end{observation}
\section{Conclusions and future work}
We introduced partial Lawvere theories and their associated notion of partial equational theory. Our definitions are guided by the appropriate universal property, replacing cartesian categories with discrete cartesian restriction categories. Knowing the right universal property determines our choice of syntax, isolating the correct class of string diagrams that replace classical terms. This enables the standard methodology of presenting a theory by means of a signature and equations, while avoiding  ad-hoc notations and eliminating the subtleties of reasoning about partial structures.

The extension is conservative: every equational theory yields a partial equational theory such that the categories of models coincide, even though our models are in $\Par$ rather than in $\Set$. The recently proposed \emph{Frobenius theories}~\cite{Bonchi2017c} take their models in the category of relations $\ct{Rel}$, and are guided by the structure of cartesian bicategories of relations~\cite{Carboni1987}. Every partial equational theory yields a Frobenius theory and again, the categories of models coincide. We feel that our notion is a sweet-spot. First, we have shown that our notion of partial theories is expressive, capturing a number of important examples. Second, we retain much of the richness of the semantic picture, via a canonical variety theorem and existence of free models.

\medskip

% \subsection{Theories as monoids}
%\begin{itemize}
%\item

% Addressing this enticing problem needs a systematic approach and a careful analysis of the formal category theory (in the sense of \cite{street1978yoneda,street1980cosmoiof,wood1982abstract}) of the 2-category of DCR categories, and in particular of its formal theory of monads \cite{Street1972}, as an obligatory step.% towards the sol
%ution of this enticing problem.
% \subsection{Tensor product of partial Lawver theories}
% \item The category of Lawvere theories can be equipped with a canonical symmetric monoidal product operation characterised by the fact that models of $\clS\otimes \clT$ are exactly the $\clS$-models in the category of $\clT$-models:%, or equivalently, the $\clT$-models in the category of $\clS$-models:
% % \[\Mod_{\clS \otimes \clT} \cong \Mod_{\clS}(\Mod_{\clT})\]
% such construction is relevant to study models of a theory interpreted in the models of another theory; say, group objects in the category of graphs, or groups in the category of groups (i.e., abelian groups).

% It thus comes as a natural question as whether the category $\ct{pLaw}$ carries a similar symmetric monoidal product; there is empirical evidence that such a construction exists, but a thorough study of all the technicalities involved is another reason to develop further the formal category theory of restriction categories.

%\item
%For computer science,
There is much future work. The fact that the syntax introduced here is inherently partial makes it well-suited to applications in computing. In particular there is an evident notion of \emph{computable model} for partial Lawvere theories, namely those models valued in the category of sets and partial recursive functions. The corresponding \emph{computable varieties} seem to be interesting for programming language semantics, and therefore worthy of study. A further step would be the lifting of this situation to a more synthetic category of computable functions, such as a \emph{Turing category}~\cite{Coc08} or \emph{monoidal computer}~\cite{Pav13}.

An important part of categorical universal algebra is played by monads, a point of view that we have not considered here.
Indeed, Lawvere theories can be seen as \emph{finitary monads}~\cite{linton}, with the category of algebras giving the associated variety. This connection has been a fruitful one, relating areas of research that are, on the surface, very different, see e.g.~\cite{cheng2020distributive,loday2012algebraic,markl2007operads}. A natural question is whether there is an analogous approach for partial algebraic theories. We conjecture that there is, with certain formal monads~\cite{Street1972} in the 2-category $\mathsf{DCRC}^\leq$ playing the role of finitary monads.
We expect that other constructions of categorical universal algebra (e.g.~\cite{freyd1966algebra,power2006countable}) will have corresponding partial accounts.

% \subsection{The internal monologue between partiality and totality: $\mathsf{Law}$ vs $\mathsf{pLaw}$}
% There is a natural functor
% \[\]
% Lawvere theories, be them sigle-sorted, multi-sorted or unsorted, identify two possible partial theories: Indeed, given a Lawvere theory we can construct
% \begin{itemize}
% 	\item[$\mathsf{Rig}$] its associated fully partial theory (ref!). This construction would associate to the theory of monoids, the theory of partial monoids.
% 	\item[$\mathsf{Par}$] its associated \textit{fake} partial theory. That is precisely the fully partial theory where we add two axioms in order to make the operations total.
% \end{itemize}

% These two constructions are graphically syntetized by the picture below.
% \[
% \begin{tikzcd}[scale=.5]
% \mathsf{Law} \arrow[rrrrrr, "\mathsf{Rig}" description, dashed, bend left] \arrow[rrrrrr, "\mathsf{Par}" description, dashed, bend right] &  & \mathsf{Prod} \arrow[ll, dashed] &  & \mathsf{Lex} \arrow[ll] &  & \mathsf{pLaw} \arrow[ll, dashed]
% \end{tikzcd}
% \]

% One possible future research direction is to understand whether these associations are functorial. Moreover, in the opposite direction a partial Lawvere theory, identified with its Karubi-completion is a category with products and thus can be naturally identified with a Lawvere theory. In the diagram above this construction is represented by that dashed unlabled arrow. Is this an adjoint triple?

\bibliography{../refs}{}
\bibliographystyle{amsalpha}

\hrulefill

%\appendix
%\input{secs/A-restriction-basics.tex}

% .
\end{document}